\documentclass[11pt, letterpaper]{article}
\usepackage{standalone}
\usepackage{authblk}
\usepackage{etoolbox}
\usepackage{tabularray} % The only package you need for advanced tables

\usepackage{nicefrac}
\usepackage[letterpaper,margin=1in]{geometry}
\usepackage{amsthm}
\usepackage{amsmath}
\usepackage{amssymb}
\usepackage{thmtools}
\usepackage{thm-restate}
\usepackage{appendix}
\usepackage{mathtools}
\usepackage{mdframed}
\usepackage{setspace}
\usepackage{tikz}
\usetikzlibrary{arrows.meta, positioning, calc}
\usepackage{algorithm}
\usepackage[linesnumbered,vlined,algo2e]{algorithm2e}
\SetKwInput{KwParams}{Parameters}
\SetAlgoSkip{bigskip}
\SetAlgoInsideSkip{smallskip}

\usepackage{placeins}

\usepackage[table]{xcolor}

\definecolor{linkcol}{rgb}{0,0,0.38}
\definecolor{citecol}{rgb}{0.8,0,0}
\definecolor{urlcol}{rgb}{0.1,0.35,0}

\usepackage[pdfpagelabels,bookmarks=false,hyperfootnotes=false,hypertexnames=false]{hyperref}
\hypersetup{colorlinks, linkcolor=linkcol, citecolor=citecol, urlcolor=urlcol}

\allowdisplaybreaks

\DeclareFontFamily{U}{BOONDOX-calo}{\skewchar\font=45 }
\DeclareFontShape{U}{BOONDOX-calo}{m}{n}{
  <-> s*[1.05] BOONDOX-r-calo}{}
\DeclareFontShape{U}{BOONDOX-calo}{b}{n}{
  <-> s*[1.05] BOONDOX-b-calo}{}
\DeclareMathAlphabet{\link}{U}{BOONDOX-calo}{m}{n}
\DeclareMathAlphabet{\blink}{U}{BOONDOX-calo}{b}{n}

\usepackage[utf8]{inputenc}

\usepackage[
backend=biber,
style=alphabetic,
citestyle=alphabetic,
maxalphanames=4,
maxcitenames=99,
mincitenames=98,
maxbibnames=99,
giveninits=true,
]{biblatex}

\usepackage[nameinlink]{cleveref}

\usepackage{lmodern}

\usepackage{bbm}
\usepackage{thm-restate}

\newtheoremstyle{light} % name
    {\topsep}                    % Space above
    {\topsep}                    % Space below
    {\itshape}                   % Body font
    {}                           % Indent amount
    {\scshape}                   % Theorem head font
    {.}                          % Punctuation after theorem head
    {.5em}                       % Space after theorem head
    {}  % Theorem head spec (can be left empty, meaning ‘normal’)

\newtheorem{theorem}{Theorem}[section]
\newtheorem{lemma}[theorem]{Lemma}

\newtheorem{definition}[theorem]{Definition}

\newtheorem{corollary}[theorem]{Corollary}

\newtheorem{observation}[theorem]{Observation}

\theoremstyle{light}

\makeatletter
\if@cref@capitalise
\crefname{claiminproof}{Claim}{Claims}
\else
\crefname{claiminproof}{claim}{claims}
\fi

\if@cref@capitalise
\crefname{algocf}{Algorithm}{Algorithms}
\else
\crefname{algocf}{algorithm}{algorithms}
\fi

\if@cref@capitalise
\crefname{conjecture}{Conjecture}{Conjectures}
\else
\crefname{conjecture}{conjecture}{conjectures}
\fi

\if@cref@capitalise
\crefname{thm}{Theorem}{Theorems}
\else
\crefname{thm}{theorem}{theorems}
\fi

\if@cref@capitalise
\crefname{lem}{Lemma}{Lemmas}
\else
\crefname{lem}{lemma}{lemmas}
\fi

\makeatother

\usepackage{url}
\usepackage{array}

\usepackage{setspace}

\usepackage{wrapfig}

\makeatletter
\newcommand{\labeltarget}[1]{\Hy@raisedlink{\hypertarget{#1}{}}}
\makeatother

\usepackage{upgreek}

\usepackage{complexity}

\usepackage[inline]{enumitem}
\usepackage{moreenum}

\usepackage[super]{nth}

\usepackage{bm}

\usepackage{xspace}

\usepackage{ifthen}

\usepackage[immediate]{silence}
\usepackage{sectsty}
\DeactivateWarningFilters[tmp]
\allsectionsfont{\boldmath}

\setlist[enumerate]{nosep,topsep=0.1em}
\setlist[enumerate,1]{label=(\roman*), leftmargin=2.2em}
\setlist[itemize]{nosep,topsep=0.3em}

\usepackage[textsize=footnotesize, color=blue!30!white]{todonotes}
\usepackage{xfrac}

\usepackage{tikz}
\usetikzlibrary{calc}
\usetikzlibrary{math}
\usetikzlibrary{shapes.geometric}
\usetikzlibrary{arrows}
\usetikzlibrary{decorations.pathreplacing, decorations.pathmorphing}

\usepackage{tcolorbox}
\tcbuselibrary{skins}

\usepackage{float}
\usepackage{graphicx}
\graphicspath{{../graphics/}}
\makeatletter
\newcommand\appendtographicspath[1]{%
  \g@addto@macro\Ginput@path{#1}%
}
\makeatother
\usepackage[margin=10pt,font=small,labelfont=bf,skip=-5pt]{caption}
\usepackage{subcaption}

\usepackage{mdframed}

\let\truehypersetup\hypersetup
\renewcommand\hypersetup[1]{}
\usepackage{bigfoot}
\let\hypersetup\truehypersetup
\csundef{poly}

\newcommand\OPT{\ensuremath{\mathrm{OPT}}}
\renewcommand{\epsilon}{\varepsilon}

\newcommand{\ALG}{\mathrm{ALG}}
\newcommand{\singleref}{\textsc{single-ref}}

\definecolor{green}{rgb}{0.4,0.85,0.6}

\newbool{sodasubm}
\setbool{sodasubm}{false}

\title{Knapsack Secretary is not $1/e$-Competitive}

\ifbool{sodasubm}{
  \author{Authors omitted}
  \date{}
}{
  \author{Marius Garbea\thanks{\href{mailto:mgarbea@drexel.edu}{mgarbea@drexel.edu}.}}
  \author{Rishi Patel\thanks{\href{mailto:riship@drexel.edu}{riship@drexel.edu}.}}
  \author{Emmanouil Pountourakis\thanks{\href{mailto:manolis@drexel.edu}{manolis@drexel.edu}.\\ The authors were partially supported by NSF award CCF 2218813.}}
  \affil{Drexel University, Department of Computer Science}
  \date{}
}

\begin{document}

\maketitle
\thispagestyle{empty}
\pagestyle{plain}

\begin{abstract}
We prove that no algorithm for the knapsack secretary problem can be $1/e$-competitive. The knapsack secretary problem was first introduced by Babaioff, Immorlica, Kempe, and Kleinberg (2007). 
There have been many improvements to the achievable competitive ratio since then, but the $1/e$ impossibility barrier has remained unchanged. Many combinatorial variants of the secretary problem, including knapsack secretary, inherit the $1/e$ impossibility by embedding the single-choice problem as a special case. We construct a family of hard instances for the $1$-$B$ knapsack secretary problem, which is a special case of the general knapsack secretary problem, to improve the existing impossibility result. We show in this special case that the competitive ratio is at most $0.36437 < \frac{1}{e} - 0.0035$. Our construction is similar to the one used by Abels, Ladewig, Schewior, and Stinzend\"orfer (2022), for which they show an impossibility of $1/(1+e)$ for ordinal algorithms,  where only the relative ranks of the items are known. Our work resolves an open question of theirs by showing that $1/e$ cannot be achieved even in the cardinal case of the $1$-$B$ knapsack secretary problem. We complement our impossibility result with a simple algorithm for $1$-$B$ knapsack secretary that is $(1/5.10-o(1))$-competitive for every fixed $B\ge 2$. This improves the guarantee obtained by applying general-purpose random-order knapsack algorithms to this special case.
\end{abstract}

\newpage
\setcounter{page}{1} 

\section{Introduction}

The secretary problem is one of the canonical models of online decision-making under uncertainty. In its classical form, there are $n$ applicants with distinct values, arriving in a uniformly random order. Upon each arrival, the decision-maker observes the applicant and must irrevocably decide whether to accept or reject them. The goal is to select the single best applicant. The optimal strategy, first analyzed by Dynkin~\cite{dynkin1963optimum}, is the familiar threshold rule: observe an initial $1/e$ fraction of the sequence, then accept the first subsequent applicant that is better than all previously seen applicants. As $n \to \infty$, this strategy succeeds with probability $1/e$, and no algorithm can do better.

The secretary problem has since become a central model in online algorithms, mechanism design, and algorithmic game theory. It is a particularly clean example of the random-order model, in which the input is adversarially chosen but then permuted uniformly at random before being presented to the algorithm. This model lies between worst-case online algorithms and fully stochastic models: it assumes no prior distribution over item values, but it rules out adversarial arrival orders. See, for example, Roughgarden~\cite{roughgarden2021beyond} for a broader discussion of random-order analysis. Secretary-type models have found applications in online auctions~\cite{hajiaghayi2004adaptive, kleinberg2005multiple}, ad allocation and online matching~\cite{mehta2010online}, and posted-price mechanism design~\cite{babaioff2007matroids, babaioff2007knapsack}.

Motivated by these applications, a large body of work studies combinatorial generalizations of the secretary problem. In such problems, the algorithm is no longer asked to select a single item, but rather a feasible set of items satisfying some combinatorial constraint. Examples include the matroid secretary problem~\cite{babaioff2007matroids}, secretary matching problems on graphs and hypergraphs~\cite{korula2009algorithms, kesselheim2013optimal, tomer2022edgearrival}, and the knapsack secretary problem~\cite{babaioff2007knapsack}. In the knapsack secretary problem, each arriving item has a value and a size, and the algorithm may accept items whose total size is at most a given capacity. The objective is to maximize the total accepted value.

A striking feature of many combinatorial secretary problems is that their strongest known impossibility often remains exactly $1/e$. This threshold is inherited from the classical secretary problem by considering instances in which the combinatorial constraint effectively allows the algorithm to choose only one relevant item. Such embeddings show that combinatorial secretary problems are at least as hard as the classical secretary problem, but they do not explain whether the combinatorial structure itself creates additional difficulty. In particular, for knapsack secretary, the traditional $1/e$ impossibility comes from a single large item competing with all other choices being irrelevant.

Breaking this $1/e$ barrier is difficult because knapsack secretary is a cardinal problem: the algorithm observes the exact numerical values and sizes of the arriving items. Cardinal information may allow the algorithm to distinguish hard cases that would be indistinguishable to an ordinal algorithm, i.e., an algorithm that only sees relative rankings among observed items. Indeed, Abels, Ladewig, Schewior, and Stinzend\"orfer~\cite{abels2022knapsack} proved an impossibility result showing that no ordinal algorithm for the knapsack secretary problem can achieve a better competitive ratio than $1/(1+e)$. However, their construction does not imply an impossibility for cardinal algorithms: on their hard instances, the numerical values themselves reveal enough information for a cardinal algorithm to perform much better.

In this paper, we prove that this obstruction can be overcome. We show that the knapsack secretary problem is not $1/e$-competitive, even for cardinal algorithms. More precisely, we construct a family of $1$-$B$ knapsack secretary instances whose optimal competitive ratio is strictly less than $1/e$. In a $1$-$B$ instance, the knapsack has capacity $B$, and each item has size either $1$ or $B$. Thus a feasible solution consists either of a single large item of size $B$, or of up to $B$ small items of size $1$. This is one of the simplest nontrivial knapsack secretary settings, but it already captures a tension absent from the classical secretary problem: the algorithm must decide whether to wait for a potentially valuable large item or start collecting small items before too much capacity-adjusted opportunity is lost.

\subsection{Breaking the $1/e$ barrier}
Our main contribution is a cardinal impossibility result below $1/e$ for the special case of $1$-$B$ knapsack secretary. The hard instances have the following structure: there are $T=B$ arrival positions; some unknown number $N$ of the items are large, each of size $B$, and the remaining $T-N$ items are small, each of size $1$. If $N>0$, then the offline optimum is the maximum-value large item. If $N=0$, then the offline optimum consists of all small items. Thus the algorithm faces a fundamental tradeoff: waiting is necessary to find the best large item when large items exist, but waiting too long is costly when no large item exists.

We show that this tradeoff yields a fraction at most
\[
    0.36437 < \frac{1}{e} - 0.0035
\]
of the best achievable competitive ratio for this family of instances. Therefore $1$-$B$ knapsack secretary, and hence general knapsack secretary, is not $1/e$-competitive.

\paragraph{An unknown-size secretary problem with a deadline.}
One way to view the hard instance is as an embedded secretary problem whose number of candidates is unknown; the candidates in our case are the large items. Since any large item has size $B$, the algorithm can accept at most one of them, and when $N>0$, the objective is exactly to select the best large item. If this were the entire problem, the uncertainty about $N$ would not by itself explain a bound below $1/e$ because $T$ is still known; depending on what $N$ is, an algorithm may do strictly better than the classical $1/e$ benchmark. The $1/e$ benchmark holds asymptotically for when $N \to \infty$.

The crucial additional ingredient is the possibility that $N=0$. In that case, the large-item secretary problem never materializes, and the optimum consists of all small items. Because accepting even one small item leaves insufficient capacity for any future large item, the first small acceptance commits the algorithm to the no-large-item branch. Conversely, every time step spent exploring for a large item permanently loses some of the small-item value available in the $N=0$ instance. We can therefore view the no-large-item case as a time limit on the exploration phase of the embedded secretary problem. Thus the algorithm must simultaneously learn whether there are any large items and, conditional on seeing large items, whether the current large item is good enough relative to future large items. These two learning tasks are coupled through the same time frame, and the small-item option imposes a deadline on both.

The impossibility result comes from balancing two opposing effects. If the decision maker encounters an instance with a large value of $N$, then the secretary component is difficult because many large candidates may compete for the maximum. If on the other hand they encounter an instance with a low value of $N$, then the no-large-item risk becomes more visible, but then identifying the maximum element becomes easier. Our family of instances exploit this tradeoff: some instances make the decision maker wait long enough to preserve value in the large-item secretary, while the degenerate no-large-item instance punishes exactly that waiting; by constructing our family to include instances with all possible $N$ between 0 and $T$, the decision maker cannot optimize for any one specific case.

\paragraph{From cardinal to ordinal.}
A central technical challenge is that our impossibility result must apply to arbitrary cardinal algorithms. The algorithm sees exact item values, so one cannot simply declare that the algorithm is ordinal and then analyze a ranked model.  This is the main distinction between our work and previous ordinal impossibility results such as the $1/(1+e)$ impossibility result of Abels, Ladewig, Schewior, and Stinzend\"orfer~\cite{abels2022knapsack}. In that work, ordinality is an exogenous restriction on the algorithm. Here, algorithms are fully cardinal; ordinality arises endogenously from the instances on which we test them.

We make this possible by constructing the hard instances so that numerical values are extremely far apart. The possible values are chosen from a rapidly growing set, for example powers separated by a large multiplicative factor. This separation ensures that, when a large item is optimal, all lower-ranked values contribute only a vanishing fraction of the optimum. Hence the only cardinally meaningful question is whether the algorithm accepts the top-ranked relevant item, not the exact numerical scale of the lower-ranked items.

We then prove that, on a sufficiently sparse infinite subsequence of such instances, the behavior of any fixed cardinal algorithm can be simulated, up to a vanishing error, by an ordinal algorithm. The reduction uses a Ramsey-theoretic pruning argument in the spirit of related reductions in online selection and prophet-inequality settings~\cite{correa2019prophet, tomer2022edgearrival}. After the reduction, the relevant state of the algorithm is only $(t,k,r)$, where $t$ is the current time, $k$ is the number of large items observed so far, and $r\in\{0,1\}$ indicates whether the current item is a record (i.e., has the highest value) among the large items seen so far. Thus, although the original problem is cardinal, our hard family forces optimal behavior to be essentially ordinal. This is the sense in which the reduction is endogenous: we choose instances on which cardinal information has no asymptotic value.

\paragraph{The distributional relaxation.}
The use of distributions is crucial to the proof. The original competitive ratio game is adversarial: first the instance family is fixed, then the algorithm is chosen, and finally the adversary selects the instance in the family that minimizes the algorithm's ratio. To upper bound the value of this game, we give the algorithm more information. We group the hard instances by the number $N$ of large items, announce a prior distribution on $N$, and analyze the optimal algorithm that knows this prior. For any fixed prior $\mathcal D$, the expected performance of the optimal prior-aware algorithm is an upper bound on the worst-case performance achievable by any prior-free algorithm on the same family. Equivalently, if even the optimal algorithm for the announced prior cannot reach $1/e$ in expectation, then no online algorithm can guarantee $1/e$ on every instance in the family.

This distributional relaxation should be interpreted as strengthening the algorithm. The algorithm is allowed to know the distribution over $N$ and choose the policy that maximizes its expected normalized payoff. After $N$ is realized, the instance still belongs to the corresponding $N$-cluster, and the Ramsey reduction lets us choose representatives of that cluster on which cardinal information is irrelevant. Thus the proof has the following structure: construct clusters indexed by $N$; reveal a prior over these clusters; characterize the optimal prior-aware policy; and then use the resulting expected value to certify that some cluster, and hence some concrete instance, is bad for every algorithm.

The prior we use is beta-binomial where $\alpha$ and $\beta$ are shape parameters:
\[
    N \sim \mathrm{BetaBin}(T,\alpha,\beta).
\]
This choice is analytically convenient because of conjugacy. Conditional on having seen $k$ large items in the first $t$ arrivals, the posterior distribution of the number of future large items $J$ remains beta-binomial:
\[
    J \sim \mathrm{BetaBin}(T-t,\alpha+k,\beta+t-k).
\]
Thus, the posterior can be updated cleanly using only the state $(t,k)$.

We also compare this beta-binomial prior with the degenerate no-large-item instance $N=0$. Before the first large item appears, these two possibilities are observationally indistinguishable to the algorithm: the only arrivals seen so far are small items. If the algorithm waits until time $\sigma$ before committing to small items, the degenerate instance gives a normalized payoff $1-\sigma$. On the other hand, the beta-binomial prior also gives a normalized payoff $F(\sigma)$ that accounts for the chance that a first large item appears before $\sigma$ and then triggers the embedded secretary problem. The adversary can therefore force the algorithm to face
\[
    \min\{F(\sigma),\,1-\sigma\}.
\]
The intersection of these two terms is the point at which the algorithm has balanced the value of continued large-item exploration against the cost of missing small items. Thus, even though the algorithm can now use the given prior to make potentially better informed decisions, we show that this advantage is not powerful enough to overcome the uncertainty between very distinct possible scenarios.

\paragraph{The dynamic program and ODE.}
Having reduced the impossibility analysis to an ordinal problem, we write the optimal prior-aware algorithm as a dynamic program. Let $N$ denote the total number of large items as before. After observing the current item at time $t$, the value function $V_t(k,r)$ is the best utility that can be achieved in a state where $t$ is the time, $k$ is the number of large items observed so far, and $r$ indicates whether the current item is a record large item. If $r=1$, then accepting the current item succeeds exactly when the current record is the best large item overall. If $r=0$, then accepting can be useful only in the case $N=0$, where the small items are optimal.

The key simplification occurs when we set $\alpha=1$ in our beta-binomial prior. Once at least one large item has appeared, the dynamic program collapses to a one-dimensional recurrence, which we denote as $V_t^{(1)}$, the best achievable utility from time $t$ onward. From this, we similarly characterize the dynamic program for the case where no large item has appeared, which we denote as $V_t^{(0)}$. We find that the dynamic programming solution converges to a limiting solution as $T \to \infty$. Since there always exists a $T$ for which the dynamic programming solution can get arbitrarily close to this limiting solution, we focus our efforts on analyzing the limiting solution.

We observe that after the first large item appears, the best possible behavior resembles a secretary rule with an endogenous $1/e$ plateau. The final loss below $1/e$ comes not from this plateau alone, but from the need to decide how long to wait before declaring that the no-large-item instance is likely enough to start accepting small items. Optimizing the resulting tradeoff over our distribution parameters gives the value $0.36437$, strictly below $1/e$.

\subsection{A refined algorithm for $1$-$B$ knapsack}

Our second contribution is algorithmic. We give a simple $(1/5.10-o(1))$-competitive algorithm for $1$-$B$ knapsack secretary for every fixed $B\ge 2$. This improves the guarantee obtained by applying general-purpose random-order knapsack algorithms to this special case.

For $B=2$, we use the boosting algorithm of Abels, Ladewig, Schewior, and Stinzend\"orfer~\cite{abels2022knapsack}, which gives a $(1/e-o(1))$-competitive algorithm for the $1$-$2$ knapsack secretary. For $B\ge 3$, we use a three-phase algorithm, in the spirit of the more general algorithm of Albers, Khan, and Ladewig~\cite{albers2021improved}. The first phase samples all items arriving before time $c_1$. In the second phase, during $[c_1,c_2)$, the algorithm attempts to select a large item, accepting the first large item whose value exceeds the best sampled value. If no large item is selected by time $c_2$, the algorithm enters the third phase, focuses on the small-item problem, and runs a $k$-secretary threshold rule with $k=B$, using the $k$-th largest small item observed before $c_2$ as the threshold.

The analysis separates the two possible offline optima. If the optimum is a large item, the algorithm selects it with probability at least
\[
    c_1 \ln\left(\frac{c_2}{c_1}\right).
\]
If the optimum consists of small items, then with probability at least $c_1/c_2$ no large item is selected before the small-item phase, and the small-item phase obtains the guarantee of the single-reference $k$-secretary algorithm of Albers and Ladewig~\cite{albers2021new}. Thus the competitive ratio is at least
\[
    \min\left\{
        c_1\ln\left(\frac{c_2}{c_1}\right),
        \frac{c_1}{c_2}\rho_B(c_2)
    \right\}
    -o(1),
\]
where
\[
    \rho_B(c)
    =
    \int_c^1
    \Pr\left[
        \operatorname{Bin}\left(2B-1,\frac{c}{t}\right)\ge B
    \right]dt.
\]
Choosing $c_1=0.3657$ and $c_2=0.6258$, and using monotonicity of $\rho_B(c)$ for $B\ge 3$, gives a uniform guarantee of $(1/5.10-o(1))$.

\paragraph{Concurrent work.} During the preparation of this manuscript, we became aware of independent concurrent work by Eric Balkanski, Jason Chatzitheodorou, Dimitris Fotakis, and Thanos Tolias~\cite{ericimpossibility2026} that also establishes an impossibility result below 1/e for the knapsack secretary problem. Although the constructions share some structural similarities, the two works use different approaches to prove their impossibility results. They also improve on the previously best-known competitive ratio for the knapsack secretary problem by presenting a $1/5.59$-competitive algorithm.
\section{Related Work}

Our work builds on several established lines of research in optimal stopping and online algorithms. We first review the classical secretary problem and its extensions to settings with an unknown number of candidates. We then discuss combinatorial generalizations, focusing on the knapsack secretary problem. Finally, we review the distinction between ordinal and cardinal algorithms, which motivates the structure of our new impossibility result for the $1$-$B$ knapsack problem.

\paragraph{The classical secretary problem.}
The classical secretary problem was introduced in the optimal stopping literature and is commonly attributed in the theoretical computer science literature to Dynkin~\cite{dynkin1963optimum}. The problem and its history are surveyed by Ferguson~\cite{ferguson1989solved}. The optimal asymptotic success probability is \(1/e\), achieved by the standard threshold rule. Gnedin's work on the game of googol~\cite{gnedin1994solution} further clarifies the relationship between cardinal and ordinal information in classical best-choice problems. In the single-choice secretary problem, cardinal values do not improve the optimal guarantee: only the relative order of observed items matters. This equivalence is special to the single-choice setting and is one reason why cardinal impossibility results become substantially more delicate in combinatorial secretary problems.

\paragraph{Unknown numbers of candidates and no-information models.}
Our hard instance can be viewed as embedding a secretary problem with an unknown number of relevant candidates: the large items are the candidates, but the algorithm does not know how many exist, and the case in which no large item exists has a different optimal solution. This connects our work conceptually to the literature on best-choice problems with an unknown number of options. Early and classical work in this direction includes Gianini and Samuels~\cite{gianini1976university}, Bruss~\cite{bruss1984unified}, Abdel-Hamid, Bather, and Trustrum~\cite{abdel1982secretary}, Bruss and Samuels~\cite{bruss1987unified,bruss1990conditions}, and Bruss's work on invariant record processes and related optimal selection problems~\cite{bruss1988invariant,bruss1987optimal}. Later work studies related last-arrival and no-information variants, including Bruss and Yor~\cite{bruss2012stochastic}, Wästlund~\cite{wastlund2011only}, Bruss~\cite{bruss20201}, and Gnedin~\cite{gnedin2021beat}. Our setting differs from these models because the arrival sequence is still a fixed-size random permutation of a worst-case instance; nevertheless, our impossibility result uses an analogous source of difficulty, namely uncertainty about whether any high-value large item exists before the algorithm must begin committing to small items.

\paragraph{Knapsack secretary and random-order packing.}
The knapsack secretary problem was introduced by Babaioff, Immorlica, Kempe, and Kleinberg~\cite{babaioff2007knapsack}, who gave one of the first constant-competitive algorithms for the problem. Subsequent work has developed stronger algorithms for knapsack and for the more general generalized assignment problem in the random-order model. Kesselheim, Tönnis, Radke, and Vöcking~\cite{kesselheim2014primal} gave improved algorithms for online packing LPs using random-order techniques. Albers, Khan, and Ladewig~\cite{albers2021improved} further improved the competitive ratios for knapsack and GAP. The current best bounds for generalized assignment and knapsack in the random-order model are due to Klimm and Knaack~\cite{klimm2025generalized}.

The special \(1\)-\(B\) knapsack secretary problem was studied explicitly by Abels, Ladewig, Schewior, and Stinzendörfer~\cite{abels2022knapsack}. They introduced a boosting approach and obtained strong guarantees for \(1\)-\(2\) knapsack secretary, as well as ordinal impossibility results for knapsack secretary. Their ordinal impossibility result shows that relative-order algorithms cannot achieve \(1/e\), but it does not rule out \(1/e\)-competitive cardinal algorithms. Our  impossibility result closes this gap differently: we construct hard instances on which cardinal information can be removed, up to vanishing error, by an instance-dependent ordinalization argument.

Our algorithmic result is also related to the \(k\)-secretary problem, where the algorithm may select up to \(k\) items. Albers and Ladewig~\cite{albers2021new} developed refined algorithms and analyses for \(k\)-secretary. In our \(1\)-\(B\) algorithm, once the algorithm commits to the small-item case, the remaining problem is exactly a \(k\)-secretary problem with \(k=B\). This allows us to use sharper \(k\)-secretary subroutines instead of general-purpose knapsack secretary algorithms.

\paragraph{Matroid, matching, and other combinatorial secretary problems.}
The matroid secretary problem, introduced by Babaioff, Immorlica, and Kleinberg~\cite{babaioff2007matroids}, asks for an online algorithm that selects an independent set in a matroid and competes with the maximum-weight independent set. The central open question is whether every matroid admits a constant-competitive secretary algorithm. Feldman, Svensson, and Zenklusen~\cite{feldman2014simple} gave an \(O(\log\log r)\)-competitive algorithm for general matroids, where \(r\) is the rank, and also developed a framework for secretary problems over intersections of matroids~\cite{feldman2018framework}. In later work, Bahrani, Beyhaghi, Singla, and Weinberg~\cite{bahrani2021formal} identified formal barriers for broad classes of simple algorithms for matroid secretary. More recently, Bérczi, Livanos, Soto, and Verdugo~\cite{berczi2025matroid} obtained improved results for matroid secretary using labeling schemes. For general downward-closed set systems, Rubinstein~\cite{rubinstein2016matroid} gives an $O(\log n \log r)$-competitive algorithm, where $r$ is the size of the largest feasible set, and Rubinstein and Singla~\cite{rubinstein2026secretary} very recently improved this to $O(\log n)$.

Secretary matching problems form another important class of combinatorial secretary problems. Korula and Pál~\cite{korula2009algorithms} studied secretary problems on graphs and hypergraphs. Kesselheim, Radke, Tönnis, and Vöcking~\cite{kesselheim2013optimal} gave an optimal online algorithm for weighted bipartite matching and extensions to combinatorial auctions. Ezra, Feldman, Gravin, and Tang~\cite{tomer2022edgearrival} established tight bounds for secretary matching in general graphs and compared the difficulty of general and bipartite graph settings. In Table~\ref{tab:problem_bounds}, we summarize the state-of-the-art for the problem classes studied.

\begin{table}[htbp]
    \centering
    \begin{tblr}{
        width = \textwidth,
        colspec = {|X[2,l,m]|X[1,l,m]|X[1,l,m]|},
        hlines,
    }
        \textbf{Problem Class} & \textbf{Best Known Ratio} & \textbf{Impossibility} \\
        General Set System (Downward-Closed) & $1/O(\log n \cdot \log r)$ \cite{rubinstein2016matroid} & $O(\log\log n / \log n)$ \cite{babaioff2007matroids} \\
        Generalized Assignment (Knapsack) & $(1 - \ln 2)/2 \approx 1/6.52$ \cite{klimm2025generalized} & $1/e-0.0035$ (this work) \\
        Graph Matching (Edge-Arrival) & $1/4$ \cite{tomer2022edgearrival} & $1/e$ \\
        Matroid & $1/O(\log \log r)$ \cite{feldman2014simple} & $1/e$ \\
    \end{tblr}
    \vspace{0.8em}
    \caption{Best known ratios and impossibility results for relevant secretary problems; Here, $r$ denotes the maximum feasible-set size for the general set system, and, respectively, the matroid rank for matroids.}
    \label{tab:problem_bounds}
\end{table}

\paragraph{Ordinal algorithms and cardinal complexity.}
A recurring theme in secretary problems is the distinction between ordinal and cardinal information. Ordinal algorithms only use comparisons between observed values, while cardinal algorithms may use the numerical values themselves. For the classical secretary problem, this distinction does not affect the optimal guarantee. For combinatorial secretary problems, however, cardinal information can be much more powerful. Gravin, Sun, and Tang~\cite{gravin2023online} systematically studied online ordinal problems, the optimality of comparison-based algorithms, and their cardinal complexity. Their work helps formalize when ordinal reasoning is without loss of generality and when cardinal information can fundamentally change the problem. Soto, Turkieltaub, and Verdugo~\cite{soto2021strong} study ordinal matroids in the secretary setting and provide new competitive ratios for multiple classes of matroids.

Our impossibility result contributes to this line of work by showing how ordinal analysis can be used to prove cardinal impossibility results, similar to the approach that Ezra, Feldman, Gravin, and Tang~\cite{tomer2022edgearrival} take to show a tight bound for secretary matching in the vertex arrival model. The point is not to restrict the algorithm to be ordinal; instead, we build a cardinal instance family whose value structure makes cardinal information asymptotically irrelevant. After a Ramsey-style reduction, every cardinal algorithm has an ordinal representative whose state depends only on time, the number of observed large items, and whether the current item is a record. This reduction is what makes the dynamic-programming and ODE analysis possible.

\paragraph{Prophet inequalities and related online selection models.}
Secretary problems are closely related to prophet inequalities, where item values are drawn from known distributions, and the algorithm competes against an offline prophet. Although the arrival and information models differ, many ideas transfer between the two areas. Correa, Dütting, Fischer, and Schewior~\cite{correa2019prophet}, for example, studied prophet inequalities for I.I.D. random variables from an unknown distribution, using techniques that are conceptually related to reducing cardinal information to comparison-based structure. Our work uses a distribution over hard instances to establish an impossibility result, rather than assuming a stochastic input model for the algorithm, but the proof similarly exploits the interaction between distributional symmetry, ordinal information, and online stopping rules.

\paragraph{Position of this work.}
Before this paper, the best known impossibility result for cardinal knapsack secretary was the classical \(1/e\) barrier obtained by embedding the single-choice secretary problem. The limitations of ordinal algorithms were known~\cite{abels2022knapsack}, but they did not preclude cardinal algorithms from using numerical values to distinguish the hard cases. This paper gives the first cardinal impossibility result below \(1/e\) for knapsack secretary. The proof replaces the classical secretary hard-instance construction with a \(1\)-\(B\) knapsack gadget that combines two competing tasks: finding the best large item when large items exist, and collecting small items when no large item exists. The resulting optimal stopping problem is captured by a beta-binomial dynamic program, its continuum limit is an explicitly solvable ODE, and the optimized bound is \(0.36437<1/e\). On the positive side, we also give a simple uniform \((1/5.10-o(1))\)-competitive algorithm for \(1\)-\(B\) knapsack secretary, showing that this basic special case admits better guarantees than those obtained by directly applying general knapsack secretary algorithms.

\section{Preliminaries}
\subsection{Knapsack Secretary Problem}\label{sec:prelim}
Let $\mathcal{I}$ be the set of items. Each item $i \in \mathcal{I}$ has a size $s_i$ and a value $v_i$. The knapsack has capacity $B$, and a feasible packing is a set $S \subseteq \mathcal{I}$ satisfying $\sum_{i \in S} s_i \le B$. The value of a packing is $v(S)=\sum_{i\in S}v_i$, and the offline optimum for a given set of items $\mathcal{I}$ is denoted by $\OPT(\mathcal{I})$.

We focus on the $1$-$B$ knapsack problem first considered by \cite{abels2022knapsack}, where every item has size either $1$ or $B$. Items of size $B$ are called \emph{large} items, and items of size $1$ are called \emph{small} items. We write $\mathcal{I}_\ell$ and $\mathcal{I}_s$ for the set of large and small items, respectively. Since a large item has size exactly $B$, any feasible packing contains either one large item or only small items.

We study the online random-order model. There are $T$ time steps, and the items arrive in a uniformly random order. When an item arrives, the algorithm observes its size and value and must irrevocably accept it, subject to feasibility, or reject it. When an arriving large item has a value that is higher than any value previously seen, we call such an item a \textit{record}. We write $\ALG(\mathcal{I})$ for the packing returned by an algorithm $\ALG$ (potentially randomized) given a set of items $\mathcal{I}$. An algorithm is $\rho$-competitive if every set of items $\mathcal{I}$,
\begin{align*}
    \mathbb{E}\left[v(\ALG(\mathcal{I}))\right] \ge \rho \cdot v(\OPT(\mathcal{I})),
\end{align*}
where the expectation is over the random arrival order $\pi$ and the randomization of the algorithm. Note when we say that an algorithm is $\rho$-competitive under an alternative objective, this means we swap $v$ for this objective.

For the hard family of instances in \Cref{sec:hard-instance}, we take $T=B$. Thus, if $N$ large items appear, then $B-N$ small items appear. In this family, once an algorithm accepts a small item, it is optimal to keep accepting small items. Throughout this work, $N$ will denote the number of large items, usually as a random variable. In \Cref{sec:cardtoord}, we show that we can study our family of instances under an alternative objective to obtain our impossibility result. This change in objective allows us to move from studying the \textit{cardinal} version of the problem to an \textit{ordinal} version where only the ranks of the items are relevant. We define this ordinal objective as follows:
\begin{align*}
    \tilde{v}(S)=
    \begin{cases}
        \mathbbm{1}\{\max_{i \in \mathcal{I}_\ell} w_i \in S\} & \text{if a large item is optimal},\\
        \frac{|S\cap \mathcal{I}_s|}{|\mathcal{I}_s|} & \text{if the small items are optimal}.
    \end{cases}
\end{align*}
We can interpret $\tilde{v}(S)$ as the fraction of optimal items in $S$. Note that under this objective, for any instance $\mathcal{I}$, $\tilde{v}(\OPT(\mathcal{I})) = 1$. If the algorithm is also given a prior $\mathcal{D}$ over instances, we denote the corresponding prior-aware algorithm by $\ALG_{\mathcal{D}}$, i.e., an algorithm that has access to the distribution. This leads to the following observation, which will help when proving our hardness result.

\begin{observation}\label{obs:ubdist}
We can upper bound the worst-case performance of the optimal online algorithm by placing a prior on the instances:
\begin{align*}
    \sup_{\ALG} \inf_\mathcal{I} \mathbb{E}[\tilde{v}(\ALG(\mathcal{I}))]
    \leq \sup_{\ALG} \mathbb{E}_{\mathcal{I}\sim \mathcal{D}}[\tilde{v}(\ALG(\mathcal{I}))]
    \leq \sup_{\ALG_{\mathcal{D}}} \mathbb{E}_{\mathcal{I}\sim \mathcal{D}}[\tilde{v}(\ALG_{\mathcal{D}}(\mathcal{I}))].
\end{align*}
where the expectation is also taken over the random arrival order and the randomization of the algorithm.
\end{observation}
Note that any time we take an expectation of an algorithm, we will assume it is also over the random arrival order and randomness of the algorithm unless stated otherwise.

\subsection{The Beta-Binomial Distribution}
\label{subsec:betabin}

The beta-binomial distribution plays a central role in our dynamic programming analysis and arises naturally in the structure of the $1$-$B$ knapsack instance. It is defined as the compound distribution obtained by first drawing a success probability $p \sim \mathrm{Beta}(\alpha, \beta)$ and then drawing $N \sim \mathrm{Binomial}(T, p)$, which is equivalent to $N \sim \mathrm{BetaBin}(T, \alpha, \beta)$. The probability mass function is
$$
\Pr[N = k] = \binom{T}{k} \frac{B(\alpha + k,\; \beta + T - k)}{B(\alpha, \beta)},
$$
where $B(\cdot, \cdot)$ denotes the beta function.

A key property that makes the beta-binomial distribution particularly useful for our analysis is its Bayesian conjugacy. In our setting, we model the unknown number of large items as drawn from a $\mathrm{BetaBin}(T, \alpha, \beta)$ prior. As the algorithm observes items arriving, the posterior distribution over the remaining number of large items remains in the beta-binomial family, with parameters updated by simple additive rules. Specifically, if the prior is $\mathrm{BetaBin}(T, \alpha, \beta)$ and the algorithm has observed $k$ large items and $t - k$ small items by time step $t$, then the posterior over the remaining number of large items is $\mathrm{BetaBin}(T - t, \alpha + k, \beta + (t - k))$. Therefore, after observing $k$ large items in the first $t$ time steps, we obtain the following posterior distribution from conjugacy for the remaining number of large items:
\begin{align*}
    N-k \mid E_{t,k} \sim \mathrm{BetaBin}(T-t,\alpha+k,\beta+t-k).
\end{align*}
where $E_{t,k}$ is the event that $k$ large items have been seen by time $t$. Equivalently, if we let $J\sim\mathrm{BetaBin}(T-t,\alpha+k,\beta+t-k)$, then
\begin{align*}
    \Pr[N=k+j\mid E_{t,k}]=\Pr[J=j].
\end{align*}
As a result, we can easily take the expectation of the number of large items $N \sim \mathrm{BetaBin}(T, \alpha, \beta)$:
\begin{align*}
    \mathbb{E}[N] = \frac{\alpha T}{\alpha + \beta}.
\end{align*}
This gives us a clean update rule that allows us to write a tractable dynamic program that we use in \Cref{sec:simplified-dp}.
\section{Improved impossibility result}

In this section, we prove our main hardness result for the $1$-$B$ knapsack secretary problem. The proof has two main ingredients. First, we construct a family of $1$-$B$ knapsack instances that vary in the number of large items $N$ and the item values. Since the values in the constructed instances are separated by large multiplicative gaps, we show that we can study these instances over the ordinal objective $\tilde{v}$ we defined in \Cref{sec:prelim}. Intuitively, the cardinal information of the item values in these instances is irrelevant because we can make the multiplicative gap between item values arbitrarily large. Thus, if the small items are optimal (i.e., $N=0$), the objective is to select as many small items as possible, and if a large item is optimal (i.e., $N>0$), the objective is to select the best large item. In \Cref{sec:cardtoord}, we show that any cardinal algorithm with access to item values under our ordinal objective can be simulated up to a vanishing loss by an ordinal algorithm whose state depends only on the time and the number of large items seen so far.

The second ingredient is a distributional relaxation. Assume that the solution of any algorithm is scored on the ordinal objective. Recall from \Cref{obs:ubdist} that we can upper bound the solution of the optimal algorithm under the worst instance using the solution of the optimal prior-aware algorithm under some fixed distribution over the instances. By evaluating the upper bound for a particular distribution, we can use this to upper bound the best achievable competitive ratio. The main benefit of using a distributional approach becomes clear when we define the optimal prior-aware algorithm through a dynamic program in \Cref{sec:dp}. For our fixed distribution, we use a beta-binomial distribution with carefully chosen parameters to vary the number of large items in our instances, which further simplifies our dynamic program as shown in \Cref{sec:simplified-dp} due to conjugacy and our choice of parameters. Intuitively, we are providing the algorithm with distributional information because it makes our analysis more tractable and can only improve the optimal competitive ratio.

From these ingredients, we are able to construct a compact dynamic program for the optimal prior-aware ordinal algorithm, which is analytically tractable. Using standard results from numerical analysis, we show that as $T$ increases, this dynamic program converges to a limiting ordinary differential equation. We solve this ODE to obtain an explicit formula for the optimal competitive ratio as a function of the time when the algorithm commits to accepting small items. This is not strong enough to show a hardness result below $1/e$, so to complete our argument, we use an adversarial approach. An adversary chooses between the beta-binomial distribution or a degenerate distribution on $N=0$. These two cases are indistinguishable to an algorithm until a large item appears. The resulting tradeoff between waiting for a large item and committing to small items gives the desired hardness result below $1/e$.

\subsection{Instance Construction}\label{sec:hard-instance}

For a given instance $\mathcal{I}$, we can represent it using a set $V$ of item values containing the values for large items and a single \textit{aggregate} value representing the total value of the small items if they form the optimal packing. This aggregate value for small items is always the minimum value in $V$. We choose to give information about the small-item values as a single aggregate value because we assume all small items have the same value in our family of instances: $\min_{v \in V} v/T$. Therefore, if a small item arrives, we assume the algorithm knows the aggregate value for when the small items form the optimal packing. Note that the aggregate value does not represent the total value of the small items when there are large items. Also note that the item labels are not relevant here because the large items look identical and the small items are defined to correspond to the minimum value in the set. This is sufficient to obtain a hardness result. If $|V| = 1$, then there are no large items and the small items form the optimal packing. Let $I(V)$ be the instance induced by the described value set $V$. Recall from \Cref{sec:prelim}, we take $T=B$. If $N$ large items appear, then $B-N$ small items appear. Let
\begin{align*}
    H_x = \{x^{2\ell}: \ell\in \mathbb{N}\},
\end{align*}
where $x>B$. An instance is then defined by a finite set $V\subseteq H_x$. Thus, $|V|=N+1$ by definition because there are $N$ large items. We define
\begin{align*}
    \mathcal{H}_x(R, m) = \left\{ V : V \subseteq H_x \cap R \text{ and } |V|=m \right\}.
\end{align*}
Here, $\mathcal{H}_x(R, m)$ represents all instances of size $m$ where $1 \leq m \leq T+1$ with the values of the instances being restricted to the set $R$. The spacing in $H_x$ ensures that when a large item is optimal, all non-optimal values are smaller by a factor of at least $x^2$. Going forward, this construction will be our family of hard instances for the $1$-$B$ knapsack secretary problem.

\subsection{Reduction from Cardinal to Ordinal}\label{sec:cardtoord}

The main result of this section shows that we can study our family of hard instances under the ordinal objective with a vanishing loss. Therefore, we essentially show that the cardinal information is irrelevant within our family which greatly simplifies our dynamic program. This is an approach similar to the one used by \cite{tomer2022edgearrival}. We also use a Ramsey-theoretic argument similar to the ones used by \cite{correa2019prophet,tomer2022edgearrival}. We start by showing that if any algorithm is $\rho$-competitive under the cardinal objective, then it must also be under the ordinal objective.

\begin{lemma}\label{lem:intermediate}
    Let $V \subseteq H_x$ and $\mathcal{I} = I(V)$ be the induced instance. Suppose that $\ALG$ is $\rho$-competitive with respect to the cardinal objective. If $|V|-1 \le B \leq x$, then the following holds:
    \begin{enumerate}
        \item\label{case:large} If the optimal solution is a large item of value $v^*=\max_{u \in V} u$, then
        \begin{align*}
            \Pr[\ALG(\mathcal{I}) \text{ contains an item of value } v^*]\ge \rho-O\left(\frac1x\right),
        \end{align*}
        \item\label{case:small} If the optimal solution is the set of small items, then
        \begin{align*}
            \frac{\mathbb{E}[|\ALG(\mathcal{I})\cap \mathcal{I}_s|]}{|\mathcal{I}_s|}\ge \rho.
        \end{align*}
    \end{enumerate}
\end{lemma}
\begin{proof}
    Consider first the large-optimum case. We can write $v^*=x^{2\ell^*}$ because $V \subseteq H_x$. Since all values in $H_x$ are separated by a factor of $x^2$, every value in $V\setminus\{v^*\}$ is at most $x^{2\ell^*-2}$. Therefore,
    \begin{align*}
        \mathbb{E}[v(\ALG(\mathcal{I}))]
        &\leq \Pr[\ALG(\mathcal{I})\text{ contains an item of value }v^*]\cdot x^{2\ell^*}
        + |V\setminus\{v^*\}|\cdot x^{2v^*-2}\\
        &\leq \Pr[\ALG(\mathcal{I})\text{ contains an item of value }v^*]\cdot x^{2\ell^*}
        + B\cdot x^{2\ell^*-2}\\
        &\leq \Pr[\ALG(\mathcal{I})\text{ contains an item of value }v^*]\cdot x^{2\ell^*}
        + x^{2\ell^*-1},
    \end{align*}
    where the last two inequalities follow from our assumption that $|V\setminus\{v^*\}| \le B \le x$. Since $\ALG$ is $\rho$-competitive and $v(\OPT(\mathcal{I}))=x^{2\ell^*}$, rearranging gives
    \begin{align*}
        \Pr[\ALG(\mathcal{I})\text{ contains an item of value }v^*]
        \ge \rho-O\left(\frac1x\right).
    \end{align*}

    In the small-optimum case, $|V|=1$, so there are no large items. Let $v_s \in V$ be the total value of the small items. We can write $v_s=x^{2\ell^*}$ once again. It follows that
    \begin{align*}
        \mathbb{E}[v(\ALG(\mathcal{I}))]
        = \frac{\mathbb{E}[|\ALG(\mathcal{I})\cap \mathcal{I}_s|]}{|\mathcal{I}_s|}\cdot x^{2\ell^*}.
    \end{align*}
    Since $v(\OPT(\mathcal{I}))=x^{2\ell^*}$, $\rho$-competitiveness implies
    \begin{align*}
        \frac{\mathbb{E}[|\ALG(\mathcal{I})\cap \mathcal{I}_s|]}{|\mathcal{I}_s|}\ge \rho.
    \end{align*}
\end{proof}

\Cref{lem:intermediate} lets us replace cardinal cardinal objective by the ordinal objective $\tilde{v}$ defined in the \Cref{sec:prelim}. We now show that after restricting $H_x$ to a particular infinite subset of values, the algorithm's decisions depend only on time, number of large items, and whether the current item is a record.

A \textit{cardinal} algorithm is fully defined by the probability it chooses to accept an item for the first time as a function of the \textit{history} $\Lambda$. This is because after an algorithm chooses to accept an item, it's remaining decisions are fixed: after accepting a large item, no other items can be accepted, and after accepting a small item, it is optimal to accept all remaining small items. The history $\Lambda$ is summarized by the current time $t$, the instance $\mathcal{I}_t$ observed so far, an indicator $r$ for whether the current item is a record (i.e., $r=1$ indicates a record, and $r=0$ otherwise), and the observed arrival order $\pi_t$ consistent with $r$.

\begin{definition}
    For a cardinal algorithm, define a decision function $d^c_t(\Lambda) = d^c_t(\mathcal{I}_t,r,\pi_t)$ as the probability that the algorithm accepts an item for the first time given a history $\Lambda$.
\end{definition}
Any time we refer to a decision function, the corresponding algorithm will be clear from context. Next, observe that the arrival order does not provide any useful information to a cardinal algorithm which will help simplify our arguments.
\begin{observation}\label{obs:arrivalorder}
    Let $d^c_t(\mathcal{I}_t,r,\pi_t)$ correspond to $\rho$-competitive cardinal algorithm. We can define a new decision function $f^c_t(\mathcal{I}_t,r,\pi_t) = \mathbb{E}_{\pi' \sim \Pi_t} [ d^c_t(\mathcal{I}_t,r,\pi')]$ over all $\mathcal{I}_t$, $r$, and $\pi_t$ where $\Pi_t$ is the distribution over all observable arrival orders at time $t$ consistent with $r$. The new cardinal algorithm defined by $f^c_t(\mathcal{I}_t,r,\pi_t)$ will still be $\rho$-competitive.
\end{observation}
Informally, the observation says that an algorithm cannot improve its competitive ratio by using arrival order information. This is clear because the payoff only depends on the instance and whether the current item is a record. Thus, we will drop $\pi_t$ from the decision functions going forward.

An \textit{ordinal} algorithm is defined similarly to a cardinal algorithm, but the history replaces the partially observed instance $\mathcal{I}_t$ with the number of large items $k$ seen so far, which we will show to be sufficiently informative for a particular family of instances. \Cref{lem:ramsey} essentially constructs an infinite subset of instances from our previously defined family where any ordinal algorithm makes the same decisions as a cardinal algorithm with a loss.
\begin{definition}\label{def:ordinalalg}
    For an ordinal algorithm, define a decision function $d^o_t(\Lambda) = d^o_t(k,r)$ as the probability that the algorithm accepts an item for the first time given a history $\Lambda$.
\end{definition}
Note that we drop the arrival order in \Cref{def:ordinalalg} from the history for the same reason as in \Cref{obs:arrivalorder}.

\begin{lemma}\label{lem:ramsey}
    For every $\epsilon>0$, there is an infinite set $R \subseteq H_x$ and a function $d^o_t(k,r)$ such that, for every $t\in[T]$, $k\in\{0,\ldots,t\}$, $r\in\{0,1\}$, and $\mathcal{I} \in \mathcal{H}_x(R, k+1)$,
    \begin{align*}
        d^c_t(\mathcal{I},r)=d^o_t(k,r)\pm O(\epsilon).
    \end{align*}
\end{lemma}
\begin{proof}
    We construct nested infinite sets
    \begin{align*}
        H_x=R_0\supseteq R_1\supseteq R_2\supseteq \cdots \supseteq R_{T+1}.
    \end{align*}
    Let $\epsilon > 0$. At each step $m\in\{1,\ldots,T+1\}$, we will construct an infinite hypergraph where the vertices are from the set $R_{m-1}$ and the hyperedges are all elements in the set $\mathcal{H}_x(R_{m-1},m)$. We will color the hyperedges using vectors. Color each hyperedge $\mathcal{I} \in \mathcal{H}_x(R_{m-1},m)$ by the vector
    \begin{align*}
        \textbf{c} = \left(\left\lfloor \frac{d^c_t(\mathcal{I},r)}{\epsilon}\right\rfloor\right)_{(t,r)\in[T]\times\{0,1\}}
    \end{align*}
    where $c_{t,r}$ will be the corresponding component of $\textbf{c}$. This is a finite coloring because there are only finitely many pairs $(t,r)$ and each probability lies in $[0,1]$. By the infinite Ramsey theorem, there is an infinite monochromatic set $R_m\subseteq R_{m-1}$ for this coloring. For each $k=m-1$, set $d^o_t(k,r) = \epsilon \cdot c_{t,r}$. Thus, we have that $d^o_t(k,r)-O(\epsilon) < d^c_t(\mathcal{I},r) \leq d^o_t(k,r)+O(\epsilon)$ for all $\mathcal{I} \in \mathcal{H}_x(R_{m},m)$. Since $R_{T+1}\subseteq R_m$ for every $m$, the same property holds on the final infinite set $R=R_{T+1}$ for all $t\in[T]$, $k\le t$, and $r \in \{0,1\}$.
\end{proof}

We are now ready to prove \Cref{thm:upperbound}, which is the main theorem of this section. This will allow us to upper bound the optimal competitive ratio of the knapsack secretary problem with the expected value of the optimal prior-aware ordinal algorithm on an ordinal objective with a specific bad distribution chosen for the instances.

\begin{lemma}\label{lem:cardinal-to-ordinal}
    Consider algorithms designed for the hard instances described in \Cref{sec:hard-instance}. Let $\ALG^C$ be a $\rho$-competitive cardinal algorithm with respect to the cardinal objective. Then there exists an ordinal algorithm $\ALG^O$ such that
    \begin{align*}
        \mathbb{E}[\tilde{v}(\ALG^O(\mathcal{I}))] \geq \rho-O\left(\frac1T\right)
    \end{align*}
    for any hard instance $\mathcal{I}$. That is, $\ALG^O$ is $(\rho - O(1/T))$-competitive with respect to the ordinal objective $\tilde{v}$.
    % In particular, by taking $u \ge T$ and then choosing $\epsilon \leq 1/T^2$, the loss can be written as $O(1/T)$.
\end{lemma}
\begin{proof}
    Since $\ALG^C$ is $\rho$-competitive with respect to the cardinal objective on our hard family of instances, we can apply \Cref{lem:intermediate} which implies that for any hard instance $\mathcal{I}$,
    \begin{align*}
        \mathbb{E}[\tilde{v}(ALG^C(\mathcal{I}))] \geq \rho - O\left(\frac{1}{x}\right).
    \end{align*}
    Now we apply \Cref{lem:ramsey} to the decision function of $\ALG^C$ to obtain a decision function $d^o_t(k,r)$ corresponding to an ordinal algorithm $\ALG^O$, i.e., the algorithm accepts an item for the first time at time $t$ having seen $k$ large items and the record indicator $r$ with probability $d^o_t(k,r)$.
    
    % Let $E_{t,k}$ be the event that $k$ large items have been seen by time $t$. This event does not assert that the current item is a record; the record information is carried separately by $r$.

    Suppose the optimal solution is a large item. The ordinal objective is then the probability of accepting the best large item. Such a payoff can occur only when the current item is a record large item. Also recall that $E_{t,k}$ is the event that we observe $k$ large items by time $t$. Let $L_t$ be the event that a best large item arrives at time $t$. Thus, we have that
    \begin{align*}
        \mathbb{E}[\tilde{v}(\ALG^C(\mathcal{I}))]
        &= \sum_{t = 1}^T\sum_{\mathcal{I}_{t}} d^c_t(\mathcal{I}_{t},1)\Pr[L_t \wedge \mathcal{I}_t \text{ is observed}]\\
        &= \sum_{t = 1}^T\sum_{k = 1}^{t}(d^o_t(k,1)\pm O(\epsilon))\Pr[L_t \wedge E_{t,k}]\\
        &= \mathbb{E}[\tilde{v}(\ALG^O(\mathcal{I}))]\pm T \cdot O(\epsilon).
    \end{align*}
    When we sum over $\mathcal{I}_t$, we mean over every possible observation of the instance $\mathcal{I}$ at time $t$. Note that the observation events form a partition for each fixed $t$ because only one realization can occur. Thus, the second sum collapses, and we are essentially summing over the probability of stopping on the best large item at time $t$. Thus, we obtain our final equality where the total additive error is multiplied by $T$.

    Now suppose the small items are optimal. Then payoff is obtained only by accepting small items. If the algorithm first accepts a small item at time $t$, it can accept the remaining small items and obtain a fraction of $(T-t+1)/T$. Therefore,
    \begin{align*}
        \mathbb{E}[\tilde{v}(\ALG^C(\mathcal{I}))]
        &=\sum_{t = 1}^T \sum_{\mathcal{I}_{t}} \frac{T-t+1}{T} d^c_t(\mathcal{I}_t,0)\Pr[\text{no large items in } \mathcal{I} \wedge \mathcal{I}_t \text{ is observed}]\\
        &=\sum_{t = 1}^T \frac{T-t+1}{T}(d^o_t(0,0)\pm O(\epsilon))\Pr[\text{no large items in } \mathcal{I} \wedge E_{t,0}]\\
        &=\mathbb{E}[\tilde{v}(\ALG^O(\mathcal{I}))]\pm \frac{T+1}{2} \cdot O(\epsilon).
    \end{align*}
    By setting $x=T$ and $\epsilon=1/T^2$, we obtain our desired result.
\end{proof}

\begin{theorem}\label{thm:upperbound}
    Consider only algorithms designed for the hard instances described in \Cref{sec:hard-instance}. We will let $\ALG^C$ denote an cardinal algorithm and $\ALG^O_\mathcal{D}$ denote an ordinal prior-aware algorithm where $\mathcal{D}$ is a distribution over instances. Let $\mathcal{D}$ be any such distribution. Then, we have that
    \begin{align*}
        \sup_{\ALG^C} \inf_\mathcal{I} \frac{\mathbb{E}\left[v\left(\ALG^C(\mathcal{I})\right)\right]}{v\left(\OPT(\mathcal{I})\right)} \leq \sup_{\ALG^O_\mathcal{D}} \mathbb{E}_{\mathcal{I} \sim \mathcal{D}}\left[\tilde{v}\left(\ALG^O_\mathcal{D}(\mathcal{I})\right)\right] + O\left(\frac{1}{T}\right).
    \end{align*}
\end{theorem}
\begin{proof}
Let the optimal competitive ratio be
\begin{align*}
    \rho^* = \sup_{\ALG^C} \inf_\mathcal{I} \frac{\mathbb{E}[v(\ALG^C(\mathcal{I}))] }{v(\OPT(\mathcal{I}))}.
\end{align*}
By \Cref{lem:cardinal-to-ordinal}, we know there exists an ordinal algorithm $\ALG^O$ such that
\begin{align*}
    \mathbb{E}[\tilde{v}(\ALG^O(\mathcal{I}))] \geq \rho^* - O\left(\frac{1}{T}\right)
\end{align*}
for any hard instance $\mathcal{I}$. It follows that
\begin{align*}
    \sup_{\ALG^O} \inf_{\mathcal{I}} \mathbb{E}[\tilde{v}(\ALG^O(\mathcal{I}))] \geq \rho^* - O\left(\frac{1}{T}\right).
\end{align*}
By \Cref{obs:ubdist}, we know that
\begin{align*}
\sup_{\ALG^O} \inf_\mathcal{I} \mathbb{E} \left[\tilde{v}\left(\ALG^O(\mathcal{I})\right)\right] \leq \sup_{\ALG^O_{\mathcal{D}}} \mathbb{E}_{\mathcal{I} \sim \mathcal{D}}\left[\tilde{v}\left(\ALG_{\mathcal{D}}^O(\mathcal{I})\right)\right].
\end{align*}
Thus, we have shown our desired result.
\end{proof}

In the remaining sections, we will analyze the upper bound proven in \Cref{thm:upperbound}. Additionally, when mentioning an algorithm, we will always be referring to a prior-aware ordinal algorithm, and the competitive ratio will always be with respect to an ordinal objective.

\subsection{Defining the Dynamic Program}\label{sec:dp}
We now construct a dynamic program to compute the competitive ratio achieved by the optimal algorithm with a given prior distribution $\mathcal{D}$ on the instances. Since an algorithm's decisions do not depend on the specific values of the items, any distribution $\mathcal{D}$ on instances can just be considered a distribution on the total number $N$ of large items instead. Define $V_t(k,r)$ as value of the dynamic program (i.e., the best achievable competitive ratio) at time $t$ given the number of large items $k$ that have been observed and an indicator $r$ for whether the current item is a record. Recall that $E_{t,k}$ is the event that $k$ large items have been observed by time $t$. The value function can be computed as the maximum between the stopping utility $y_t(k,r)$ and the expected continuation utility $\mathbb{E}_{N\sim\mathcal{D}}[C_t(k,N)\mid E_{t,k}]$ when $t \geq 1$:
\begin{align*}
    V_t(k,r)=\max\left\{y_t(k,r),\mathbb{E}_{N\sim\mathcal{D}}[C_t(k,N)\mid E_{t,k}]\right\}.
\end{align*}
The stopping payoff, i.e., the payoff if we choose to accept an item for the first time, is
\begin{align*}
    y_t(k,r)=
    \begin{cases}
        \frac{T-t+1}{T}\Pr_{N\sim\mathcal{D}}[N=0\mid E_{t,k}] & r=0\\
        \mathbb{E}_{N\sim\mathcal{D}}\left[\frac{k}{N}\,\middle|\,E_{t,k}\right] & r=1
    \end{cases}
\end{align*}
where the $r=1$ case indicates that the current item is a record and the $r=0$ case indicates the current item is not a record. When $r=1$, the utility we obtain is the probability that the current record is the overall best item:
\begin{align*}
    \Pr[\text{current item is overall best}\mid \text{current item is relative best}]
    =\frac{1/N}{1/k}=\frac{k}{N}.
\end{align*}
Note that if $r=1$, it must be the case that $N > 0$. When $r=0$, we only obtain utility if the current item is a small item and no large items arrive. In this case, the utility we obtain is exactly the fraction of small items we can still accept.

The continuation utility $C_t(k,n)$ is the utility we obtain from not accepting at time $t$ conditioned on $E_{t,k}$ and $N = n$. To compute it, we consider every state that is possible from an item arriving at time $t+1$. There will be $n-k$ large items and $T-t-(n-k)$ small items remaining. If the next item is large, then by rank symmetry it is a record with probability $1/(k+1)$ and not a record with probability $k/(k+1)$. Hence,
\begin{align*}
    C_t(k,n)
    ={}& \frac{n-k}{(k+1)(T-t)}V_{t+1}(k+1,1)
    +\frac{k(n-k)}{(k+1)(T-t)}V_{t+1}(k+1,0)+\frac{T-t-(n-k)}{T-t}V_{t+1}(k,0).
\end{align*}
The first term is when we transition to a state where the next item is a record, the second term is when transition to a state where the next item is large but not a record, and the last term is when we transition to a state where the next item is small. Finally, we write the boundary conditions which occur at the last timestep:
\begin{align*}
    V_T(k,1)&=1,\\
    V_T(k,0)&=
    \begin{cases}
        1/T & k=0\\
        0 & k\ne 0
    \end{cases}.
\end{align*}
Before any item is observed, we define the optimal value (i.e., the starting point of the DP) as
\begin{align*}
    V_0=\mathbb{E}_{N\sim\mathcal{D}}\left[\frac{N}{T}V_1(1,1)+\frac{T-N}{T}V_1(0,0)\right].
\end{align*}

\subsection{Simplifying the Dynamic Program}\label{sec:simplified-dp}
In this section, we use the beta-binomial distribution with $\alpha=1$ for the number of large items to significantly simplify the dynamic program. This choice is used for the rest of the section and removes the dependence on $k$ for the case where at least one large item has appeared. Here, $V_t^{(1)}$ represents the best achievable competitive ratio conditioned on having seen at least one large item by time $t$. Note that since $k \ge 1$, the small items and non-record large items do not give any utility and will be treated the same.

\begin{lemma}\label{lem:dp-large-simplified}
    Let $\mathcal{D}=\mathrm{BetaBin}(T,1,\beta)$. Conditional on having seen at least one large item, the optimal competitive ratio is described by the one-dimensional dynamic program for $t \geq 1$:
    \begin{align*}
        V_t^{(1)}=
        \begin{cases}
            \max\left\{\frac{\beta+t}{\beta+T},C_t\right\} & t<T\\
            1 & t=T
        \end{cases}
    \end{align*}
    where
    \begin{align}\label{eq:dynprog}
        C_t=\begin{cases}
            \frac{1}{\beta+t+1}V_{t+1}^{(1)}+\frac{\beta+t}{\beta+t+1}C_{t+1} & t<T\\
            0 & t=T
        \end{cases}.
    \end{align}
\end{lemma}
\begin{proof}
    Since the beta-binomial distribution parameter $\alpha=1$, the probability that the next item is large conditioned on $E_{t,k}$ is
    \begin{align*}
        \frac{\alpha + k}{(\alpha + k) + (\beta + t - k)} = \frac{k+1}{\beta+t+1}.
    \end{align*}
    If the next item is large, rank symmetry implies that it is a record with probability $1/(k+1)$. Therefore, the probability that the next item is a record large item conditioned on $E_{t,k}$ is
    \begin{align*}
        p_t = \frac{k+1}{\beta+t+1} \cdot \frac{1}{k+1} = \frac{1}{\beta+t+1}.
    \end{align*}
    The probability of transitioning to a non-record state, i.e., the next item to arrive is a non-record large item or a small item, occurs with the remaining probability $1-p_t$.
    
    The stopping utility when the current item is a record conditioned on $E_{t,k}$ also loses its dependence on $k$:
    \begin{align*}
        \mathbb{E}\left[\frac{k}{N}\,\middle|\,E_{t,k}\right]
        &=\Pr[\text{no future item is a record} \mid E_{t,k}]\\
        &=\prod_{j=t}^{T-1}\left(1-p_t\right)
        =\frac{\beta+t}{\beta+T}
    \end{align*}
    where the product telescopes. Let $V_t^{(1)}$ be the optimal competitive ratio achievable at time $t$ when the current item is a record, and let $C_t$ be the continuation utility at time $t$. The dynamic program can only stop when the current item is a record; otherwise, it will choose to continue. Hence, the structure of the dynamic program follows from the probability $p_t$ of a record arriving and the stopping utility. The boundary conditions are immediate at $t=T$. Thus, for $k\ge1$, we have shown that the optimal competitive ratio is described by a dynamic program that only depends on the time $t$.
\end{proof}

We now take the dynamic program for the case where at least one large item has appeared from \Cref{lem:dp-large-simplified} and use it to construct a dynamic program for the case where no large item has appeared by time $t$.
\begin{lemma}\label{lem:dp-no-large-discrete}
    Let $\mathcal{D}=\mathrm{BetaBin}(T,1,\beta)$ and $p_t=\frac{1}{\beta+t+1}$. Conditional on having seen no large items, the optimal competitive ratio is described by the following dynamic program for $t \geq 1$:
    \begin{align*}
        V_t^{(0)}=\max\left\{\left(\frac{\beta+t}{\beta+T}\right)\left(\frac{T-t+1}{T}\right),\ p_tV_{t+1}^{(1)}+(1-p_t)V_{t+1}^{(0)}\right\}.
    \end{align*}
    Moreover, the competitive ratio before the process begins can be written as
    \begin{align}\label{eq:no-item-val}
        V_0^{(0)} = \max_\mathcal{\mathcal{T}} \underset{\tau \sim \mathcal{T}}{\mathbb{E}} \left[ p_0V_{1}^{(1)} + \left(\frac{\beta }{\beta + T}\right) \left(\frac{T - \tau + 1}{T}\right) + \sum_{t=1}^{\tau-1} \left(\frac{\beta}{\beta + t}\right) p_t V_{t+1}^{(1)} \right].
    \end{align}
    where $\mathcal{T}$ is a distribution over $[T]$ and $\tau$ is a random variable representing when small items start being accepted if no large items have been seen.
\end{lemma}
\begin{proof}
    Similar to \Cref{lem:dp-large-simplified}, we have that the probability that the next item is a record large item conditioned on $E_{t,0}$ is
    \begin{align*}
        p_t = \frac{1}{\beta + t +1}
    \end{align*}
    If no large item has appeared by time $t$ and the algorithm accepts the current small item, then it should continue accepting small items. Thus, $\frac{T-t+1}{T}$ fraction of small items will be accepted. Note that the first large item to arrive is always a record. Thus, the probability that no large item appears in the instance conditioned on $E_{t,0}$ is
    \begin{align*}
        \Pr[N=0\mid E_{t,0}] &= \Pr[\text{no future item is a record} \mid E_{t,0}]\\
        &=\prod_{j=t}^{T-1}\left(1-p_t\right)
        =\frac{\beta+t}{\beta+T}.
    \end{align*}
    We only obtain utility if the small items are optimal, so the expected stopping utility will be the product of the fraction of small items accepted and the probability that no future item is a record. If the algorithm continues, then a record appears at time $t+1$ with probability $p_t$; otherwise, the process remains in a state where no large items have been seen. Hence, the recurrence for $V_t^{(0)}$ follows.

    To write the optimal competitive ratio before the process begins in the stopping-time form, we expand the recurrence. In the first timestep, either a large item or a small item arrives. With probability $p_0$, a large item arrives and we obtain utility $V_1^{(1)}$, and with probability $1-p_0$, a small item arrives. Note $\tau$ is the time at which we start accepting small items if no large items have arrived. Thus, we only need to expand $V_{1}^{(0)}$ up to time $\tau$. The terminal term is the product of the payoff $\frac{T-\tau+1}{T}$ obtained when stopping at time $\tau$ and the probability that no large item has arrived at any time. The $\frac{\beta}{\beta+t} p_t$ term is the probability that no large item has appeared by time $t$ and then at time $t+1$, the first large item arrives and gives utility $V_{t+1}^{(1)}$.
\end{proof}

\subsection{Convergence of the Dynamic Program to the Limiting ODE}\label{sec:dp-to-ode}

In this section, we will show that $V_t^{(1)}$ converges to an ordinary differential equation as $T \to \infty$. By solving the differential equation, we obtain a closed-form solution which is the \textit{limiting} ODE for $V_t^{(1)}$. From this, we can also obtain the limiting value of $V_0^{(0)}$ as $T \to \infty$ by substituting in our closed-form solution. Since our dynamic programming solution converges to this limiting value, there exists a sufficiently large $T$ such that we get arbitrarily close to this value.

As $T \to \infty$, the problem changes into a continuous-time version of the problem. Thus, we will reparameterize our dynamic program so we can take its limit with respect to $T$. Let $\beta=\gamma T$, $s=t/T$, and $h = 1/T$. Here, $s$ is the time and $h$ is the step size where both have been scaled to be on $[0,1]$. We will show that the limiting differential equation for the continuation utility in \Cref{lem:dp-large-simplified} is
\begin{align}\label{eq:large-ode}
    \frac{d C_s}{d s}
    =\frac{1}{\gamma+s}\left(C_s-\max\left\{\frac{\gamma+s}{\gamma+1},C_s\right\}\right),\ C_1=0.
\end{align}
using tools from numerical analysis. Define
\begin{align*}
    f(s,C)=\frac{1}{\gamma+s}\left(C-\max\left\{\frac{\gamma+s}{\gamma+1},C\right\}\right).
\end{align*}
Then, $\frac{dC_s}{ds}=f(s,C_s)$ with terminal condition $C_1=0$. We will show that $f(s, C)$ is Lipschitz in $C$ for all $s \in [0,1]$. Additionally, it will be useful to show that $C_s$ (or $C(s)$) is Lipschitz in $s$ where $s \in [0,1]$. This will allow us to show that numerical methods such as the Euler method will converge to the solution of the ordinary differential equations. As we will see from \Cref{lem:backward-euler-dp}, our dynamic program can be viewed as the backward Euler method for the limiting ODE \eqref{eq:large-ode}, a concept from numerical analysis.

\begin{lemma}\label{lem:ode-lipschitz}
    The function $f(s, C)$ is Lipschitz on $C$ for all $s \in [0,1]$.
\end{lemma}
\begin{proof}
    For $C_1,C_2\in\mathbb{R}$,
    \begin{align*}
        |f(s,C_1)-f(s,C_2)|
        &= \frac{1}{\gamma+s}\left|\left(\frac{\gamma+s}{\gamma+1}-C_1\right)^+-\left(\frac{\gamma+s}{\gamma+1}-C_2\right)^+\right|\\
        &\le \frac{1}{\gamma+s}|C_1-C_2|
        \le \frac1\gamma |C_1-C_2|
    \end{align*}
    where $x^+ = \max\{x,0\}$.
\end{proof}

\begin{lemma}\label{lem:s-lipschitz}
    Let $C(s)$ be the solution to ODE \eqref{eq:large-ode}. Then $C(s)$ is Lipschitz on $s$ where $s \in [0,1]$.
\end{lemma}
\begin{proof}
    We know that $C(s) \in [0,1]$ because it represents a competitive ratio. Let $s_1,s_2 \in [0,1]$ where $s_1 < s_2$. We also know that $C(s)$ is differentiable on $(0,1)$, so we can apply the mean value theorem. Thus, there exists a point $u \in (s_1,s_2)$ such that 
    \begin{align*}
        f(u,C(u)) = \frac{C(s_2) - C(s_1)}{s_2-s_1}.
    \end{align*}
    By rearranging and applying the absolute value to both sides, we obtain the following:
    \begin{align*}
        | C(s_2) - C(s_1) | = |f(u,C(u))| \cdot |s_2-s_1|.
    \end{align*}
    Observe that $| f(u,C(u)) | \leq \frac{1}{\gamma}$ because $\left| C(s)-\max\left\{\frac{\gamma+s}{\gamma+1},C(s)\right\} \right| \leq 1$. Thus, we have that
    \begin{align*}
        | C(s_2) - C(s_1) | \leq \frac{1}{\gamma} \cdot |s_2-s_1|.
    \end{align*}
\end{proof}

\begin{lemma}\label{lem:backward-euler-dp}
    Computing an approximate value for ODE \eqref{eq:large-ode} using the backward Euler method is equivalent to computing recurrence \eqref{eq:dynprog} for the dynamic program in \Cref{lem:dp-large-simplified}.
\end{lemma}
\begin{proof}
    Let $h=1/T$ be the step size. Recall that $\beta=\gamma T$ and $s=t/T$. Applying the backward Euler method to ODE \eqref{eq:large-ode} starting from $s=1$ gives
    \begin{align*}
        C_s
        &=C_{s+h}-h\, f(s+h,C_{s+h})\\
        &=C_{s+h}-\frac{h}{\gamma+s+h}\left(C_{s+h}-\max\left\{\frac{\gamma+s+h}{\gamma+1},C_{s+h}\right\}\right).
    \end{align*}
    By substituting $s$, $h$, and $\gamma$, we return to the original discrete indices. This gives recurrence \eqref{eq:dynprog} from the dynamic program in \Cref{lem:dp-large-simplified}:
    \begin{align*}
        C_t
        &=\frac{\beta+t}{\beta+t+1}C_{t+1}
        +\frac{1}{\beta+t+1}\max\left\{\frac{\beta+t+1}{\beta+T},C_{t+1}\right\}.
    \end{align*}
    where $V_{t+1}^{(1)} = \max\left\{ \frac{\beta+t+1}{\beta+T},C_{t+1} \right\}$. The base cases clearly follow. Thus, we have our desired result.
\end{proof}

We can use the dynamic program from \Cref{lem:dp-large-simplified} and create a piecewise-linear interpolation, i.e., at each time $t$, evaluate the dynamic program and form a linear interpolation between successive points at $t$ and $t+1$. We will show that the piecewise-linear interpolation converges to ODE \eqref{eq:large-ode} as the step size decreases. Going forward, we will reparameterize the dynamic program to use $s$, $h$, and $\gamma$ as was shown in \Cref{lem:backward-euler-dp}. To avoid confusion, we will use $D_s$ to denote the piecewise-linear interpolation evaluated at $s$. Note that this is the same as evaluating the reparameterized dynamic program at time $s$. Let $C(s)$ be the solution of ODE \eqref{eq:large-ode} evaluated at $s$.

\begin{theorem}\label{thm:dp-converges-ode}
    Let $D$ be the piecewise-linear interpolation formed from the continuation utility of the dynamic program in \Cref{lem:dp-large-simplified}. Then $D$ converges uniformly on $[0,1]$ to the solution $C$ of ODE \eqref{eq:large-ode} as $h \to 0$.
\end{theorem}
\begin{proof}
    To show that $D$ converges to $C$ on $[0,1]$ as $h \to 0$, we will compute the total accumulated error at any time $s$ and show that it approaches $0$ as $h \to 0$. Define $\varepsilon_s = D_s - C(s)$. By definition, we have
    \begin{align*}
        D_s & = D_{s+h} - h f(s+h, D_{s + h}),\\
        C(s) & = C(s + h) - \int_{s}^{s + h} f(u, C(u)) du.
    \end{align*}
    It follows that
    \begin{align*}
        \varepsilon_s =\ & \varepsilon_{s + h} + \int_{s}^{s + h} f(u, C(u)) du - h f(s+h, D_{s + h})\\
        =\ & \varepsilon_{s + h} + \int_{s}^{s + h} \left[f(u, C(u)) - f(s+h, C(s+h) \right]du - h \left[f(s+h, D_{s + h}) - f(s+h, C(s+h))\right]
    \end{align*}
    We take the absolute value of both sides of the equation and obtain the following upper bound on the accumulated error by applying \Cref{lem:ode-lipschitz}:
    \begin{align*}
        |\varepsilon_s| \leq\ & |\varepsilon_{s + h}| + \left| \int_{s}^{s + h} \left[f(u, C(u)) - f(s+h, C(s+h) \right]du \right| + h | f(s+h, D_{s + h}) - f(s+h, C(s+h)) |\\
        \leq\ &  |\varepsilon_{s + h}| + \left| \int_{s}^{s + h} \left[f(u, C(u)) - f(s+h, C(s+h) \right]du \right| + \frac{h}{\gamma} | \varepsilon_{s+h} |.
    \end{align*}
    It remains to upper bound the integral. Observe that $C(s)$ is bounded and Lipschitz on $[0,1]$ by \Cref{lem:s-lipschitz}. Additionally, observe that $\frac{1}{\gamma + s}$ and $\frac{\gamma + s}{\gamma + 1}$ are bounded and Lipschitz on $[0,1]$. Since the family of Lipschitz functions bounded on $[0,1]$ are closed under the $\max$ operation, difference operation, and product operation, it follows that $s \mapsto f(s,C(s))$ is Lipschitz on $s$ where $L$ is the corresponding Lipschitz constant. Thus, we can bound the integral as follows:
    \begin{align*}
        \left| \int_{s}^{s + h} \left[f(u, C(u)) - f(s+h, C(s+h) \right]du \right| & \leq \int_{s}^{s + h} L|u - (s + h)| du\\
        & = \left[ L\left((s+h)u - \frac{u^2}{2}\right) \right]^{u=s+h}_{u=s}\\
        & = \frac{Lh^2}{2}
    \end{align*}
    It follows that
    \begin{align*}
        |\varepsilon_s| \leq \left(1 + \frac{h}{\gamma}\right) |\varepsilon_{s + h}| + \frac{Lh^2}{2} = \left(1 + \frac{h}{\gamma}\right) |\varepsilon_{s + h}| + O(h^2).
    \end{align*}
    We will unroll the recurrence to obtain the following:
    \begin{align*}
        |\varepsilon_s| \leq \left(1+\frac{h}{\gamma}\right)^{\frac{1}{h}(1-s)} |\varepsilon_1| + \left[ \sum_{k=0}^{\frac{1}{h}(1-s-h)} \left(1 + \frac{h}{\gamma}\right)^k \right] O\left(h^2\right)
    \end{align*}
    Since $e_1=0$, we have that
    \begin{align*}
        |e_t| \leq O\left(h^2\right) \sum_{k=0}^{\frac{1}{h}(1-s-h)} \left(1 + \frac{h}{\gamma}\right)^k \leq O\left(h^2\right) \frac{1}{h} \left(1 + \frac{h}{\gamma}\right)^\frac{1}{h} \leq O\left(h\right) e^{1/\gamma} = O\left(h\right).
    \end{align*}
    The last inequality follows from the fact that $(1+x) \le e^x$. Thus, the accumulated error approaches $0$ as $h \to 0$ proving uniform convergence.
\end{proof}

We obtain the following corollary which says our dynamic program can get arbitrarily close to the solution of the ODE as we decrease the step size. This allows us to focus on analyzing the solution of the ODE.
\begin{corollary}\label{cor:converge}
    For every $\varepsilon > 0$, there exists a step size $h$ such that $| D_0 - C(0) | < \varepsilon$.
\end{corollary}

\subsection{Solving the ODE}\label{sec:solving-ode}
In this section, we solve the system defined by $V_s^{(1)}=\max\left\{\frac{\gamma+s}{\gamma+1},C_s\right\}$ and ODE \eqref{eq:large-ode} for the continuation utility $C_s$. Note that going forward, we will only be discussing the continuous formulation of the problem. We will use standard tools from ordinary differential equations.

\begin{lemma}\label{lem:v1-ode-solution}
    Suppose $0<\gamma\le \frac{1}{e-1}$, and let $s^*=\frac{1+\gamma}{e}-\gamma$.
    Then $s^*\in[0,1]$, and solving ODE \eqref{eq:large-ode} gives the following solution for $V_s^{(1)}$:
    \begin{align*}
        V_s^{(1)}=\max\left\{\frac{\gamma+s}{\gamma+1},\frac{1}{e}\right\}
        =
        \begin{cases}
            \frac{1}{e} & s\le s^*\\
            \frac{\gamma+s}{\gamma+1} & s>s^*
        \end{cases}.
    \end{align*}
\end{lemma}
\begin{proof}
    The condition $s^*\ge0$ is equivalent to $\gamma \le \frac{1}{e-1}$, and $s^*\le1$ follows from $\gamma\ge0$. For $s > s^*$, stopping weakly dominates continuing, so $V_s^{(1)}=\frac{\gamma+s}{\gamma+1}$ and ODE \eqref{eq:large-ode} can be written as
    \begin{align*}
        \frac{d C_s}{ds}=\frac{1}{\gamma+s}\left(C_s-\frac{\gamma+s}{\gamma+1}\right),
        \qquad C_1=0.
    \end{align*}
    Multiplying both sides of the differential equation by an integrating factor $\frac{1}{\gamma+s}$ gives
    \begin{align*}
        \frac{d}{d s}\left(\frac{C_s}{\gamma+s}\right)
        =-\frac{1}{(\gamma+1)(\gamma+s)}.
    \end{align*}
    We integrate both sides from $s$ to $1$ and rearrange yielding the following:
    \begin{align*}
        \int_s^1 \frac{d}{d u}\left(\frac{C_u}{\gamma+u}\right) du
        & =- \frac{1}{\gamma+1} \int_s^1 \frac{1}{\gamma+u} du,\\
        C_s & =\frac{\gamma+s}{\gamma+1}\ln\left(\frac{\gamma+1}{\gamma+s}\right).
    \end{align*}
    Note that $C_1 = 0$, so the solution is valid. The indifference point, i.e., the time when the utility from stopping is the same as continuing, is determined by solving for $s^*$ in $C_{s^*}=\frac{\gamma+s^*}{\gamma+1}$. Rearranging gives us
    \begin{align*}
        \ln\left(\frac{\gamma+1}{\gamma+s^*}\right)=1.
    \end{align*}
    This gives us the following:
    \begin{align*}
        s^*=\frac{1+
        \gamma}{e}-\gamma.
    \end{align*}
    Observe that $\frac{\gamma+s^*}{\gamma+1}=\frac{1}{e}$. For $s \leq s^*$, continuing weakly dominates stopping, so $\frac{d C_s}{ds}=0$. Since $C_{s^*}=\frac{1}{e}$, it is clear that $C_s = \frac{1}{e}$ for all $s \leq s^*$. This proves the claimed solution for $V_s^{(1)}$.
\end{proof}

\subsection{The Case of No Large Items}\label{sec:no-large-case}
In this section, we construct a continuous analogue for $V_0^{(0)}$ defined by Equation~\eqref{eq:no-item-val} in \Cref{lem:dp-no-large-discrete} similar to how we did for $V_t^{(1)}$. From this, we analyze $F_{\gamma}(\sigma)$ which is our utility if we start accepting small items at time $\sigma \in [0,1]$ where $\sigma = \tau/T$ with parameter $\gamma$ for the distribution. To further evaluate $F_{\gamma}(\sigma)$, we use our solution for $V_s^{(1)}$ from \Cref{sec:solving-ode}.

\begin{lemma}
    The continuous analogue of Equation~\eqref{eq:no-item-val} in \Cref{lem:dp-no-large-discrete} is
    \begin{align*}
        V_0^{(0)}=\sup_\mathcal{\mathcal{S}} \underset{\sigma \sim \mathcal{S}}{\mathbb{E}} \left[ F_\gamma(\sigma) \right]
    \end{align*}
    where
    \begin{align*}
        F_\gamma(\sigma) =\int_0^\sigma \frac{\gamma}{(\gamma+s)^2}V_s^{(1)}\,ds + \frac{\gamma}{\gamma+1}(1-\sigma)
    \end{align*}
    and $\mathcal{S}$ is a distribution over $[0,1]$.    
\end{lemma}
\begin{proof}
    To see why this is the continuous analogue, we will reparameterize the argument of the expectation in Equation~\eqref{eq:no-item-val} by letting $\beta = \gamma T$, $s=t/T$, $h=1/T$, and $\sigma = \tau/T$ similar to before:
    \begin{align*}
        & p_0V_{1}^{(1)} + \left(\frac{\beta }{\beta + T}\right) \left(\frac{T - \tau + 1}{T}\right) + \sum_{t=1}^{\tau-1} \left(\frac{\beta}{\beta + t}\right) p_t V_{t+1}^{(1)}\\
        =\ & \frac{h}{\gamma + h} V_{h}^{(1)} + \left(\frac{\gamma}{\gamma + 1}\right)\left(1-\sigma+h\right) + \sum_{\substack{t\in[\tau-1]\\s=th}} \left(\frac{\gamma}{\gamma + s}\right)\left(\frac{h}{\gamma + s + h}\right) V_{s+h}^{(1)}.
    \end{align*}
    By taking the limit as $h \to 0$, we obtain the desired result:
    \begin{align*}
        & \lim_{h \to 0} \left[ \frac{h}{\gamma + h} V_{h}^{(1)} + \left(\frac{\gamma}{\gamma + 1}\right)\left(1-\sigma+h\right) + \sum_{\substack{t\in[\tau-1]\\s=th}} \left(\frac{\gamma}{\gamma + s}\right)\left(\frac{1}{\gamma + s + h}\right) V_{s+h}^{(1)} h \right]\\
        =\ & \int_{0}^{\sigma} \frac{\gamma}{(\gamma + s)^2} V_{s}^{(1)} ds +  \left(\frac{\gamma}{\gamma + 1}\right)\left(1-\sigma\right).
    \end{align*}
    In the continuous version, $\sigma$ is a random variable representing the time on $[0,1]$ at which the algorithm starts accepting small items if no large item has appeared yet. The first term accounts for the event that the first large item appears before $\sigma$, and the second term accounts for the event that the algorithm reaches time $\sigma$ having observed no large items and starts accepting small items.
\end{proof}

\begin{lemma}\label{lem:F-closed-form}
    For $0<\gamma\le\frac{1}{e-1}$ and $s^*=\frac{1+\gamma}{e}-\gamma$,
    \begin{align*}
        F_\gamma(\sigma)=
        \begin{cases}
            \frac1e\left(\frac{\sigma}{\gamma+\sigma}\right)+\frac{\gamma}{\gamma+1}(1-\sigma), & \sigma\le s^*,\\
            \frac1e\left(\frac{s^*}{\gamma+s^*}\right)
            +\frac{\gamma}{\gamma+1}\ln\left(\frac{\gamma+\sigma}{\gamma+s^*}\right)
            +\frac{\gamma}{\gamma+1}(1-\sigma), & \sigma>s^*.
        \end{cases}
    \end{align*}
    Moreover, $F_\gamma(\sigma)$ is concave on $[0,1]$, so $\sigma$ will be taken to be deterministic.
\end{lemma}
\begin{proof}
    Substitute the expression for $V_s^{(1)}$ from \Cref{lem:v1-ode-solution}. If $\sigma\le s^*$, then $V_s^{(1)}=\frac{1}{e}$ throughout the integral, so
    \begin{align*}
        F_\gamma(\sigma)
        &=\int_0^\sigma \frac{\gamma}{(\gamma+s)^2}\cdot\frac1e\,ds
        +\frac{\gamma}{\gamma+1}(1-\sigma)\\
        &=\frac1e\left(\frac{\sigma}{\gamma+\sigma}\right)+\frac{\gamma}{\gamma+1}(1-\sigma).
    \end{align*}
    If $\sigma>s^*$, split the integral at $s^*$ and use $V_s^{(1)}=\frac{\gamma+s}{\gamma+1}$ on $[s^*,\sigma]$ to obtain the second expression:
    \begin{align*}
        F(\sigma) & = \int_0^{s^*} \frac{\gamma}{(\gamma + s)^2} \cdot \frac{1}{e} ds + \int_{s^*}^{\sigma} \frac{\gamma}{(\gamma + s)} \cdot \frac{1}{\gamma+1} ds + \frac{\gamma}{\gamma + 1} (1-\sigma)\\
        & = \frac{1}{e}\left(\frac{s^*}{\gamma + s^*}\right) + \frac{\gamma}{\gamma + 1}\left[\ln\left(\frac{\gamma + \sigma}{\gamma + s^*}\right)\right] + \frac{\gamma}{\gamma + 1} (1-\sigma)
    \end{align*}
    The first derivatives and second derivatives are
    \begin{align*}
        F_\gamma'(\sigma)&=\frac{\gamma}{e(\gamma+\sigma)^2}-\frac{\gamma}{\gamma+1}, &
        F_\gamma''(\sigma)&=-\frac{2\gamma}{e(\gamma+\sigma)^3}
        \qquad (\sigma \leq s^*),\\
        F_\gamma'(\sigma)&=\frac{\gamma}{\gamma+1}\left(\frac{1}{\gamma+\sigma}-1\right), &
        F_\gamma''(\sigma)&=-\frac{\gamma}{(\gamma+1)(\gamma+\sigma)^2}
        \qquad (\sigma>s^*).
    \end{align*}
    Both second derivatives are negative, and the one-sided first derivatives agree at $s^*$. Hence $F_\gamma(\sigma)$ is concave. Maximizing the expected value of this concave function over a randomized $\sigma$ is no better than maximizing over a deterministic $\sigma$. To see this, observe that $\mathbb{E}_\sigma[F_\gamma(\sigma)] \leq F_\gamma(\mathbb{E}_\sigma[\sigma])$ by Jensen's inequality.
\end{proof}

\subsection{Analytic Upper Bound Argument}\label{sec:analytic-upper-bound}
We now use the preceding characterization to obtain an upper bound below $1/e$. We use an adversarial approach where an adversary chooses between two distributions: the beta-binomial distribution with $\alpha=1$ and $\beta=\gamma T$ and the degenerate distribution with $N=0$. Before the algorithm sees a large item, these two cases are indistinguishable. If the algorithm waits until time $\sigma$ before accepting small items, then the degenerate no-large distribution gives payoff $1-\sigma$, while the beta-binomial distribution gives payoff $F_\gamma(\sigma)$. By the concavity property of $F_\gamma(\sigma)$ proven in \Cref{lem:F-closed-form} and $1-\sigma$ being linear, the optimal $\sigma$ can be taken to be deterministic. Thus, the relevant guarantee to optimize over $\sigma$ is $\min\{F_\gamma(\sigma),1-\sigma\}$ because the adversary chooses the worst case. We will show there exists a $\gamma$ such that the intersection of $F_\gamma(\sigma)$ and $1-\sigma$ optimizes the guarantee and that this guarantee is below $1/e$. Throughout this section, we will let $s^* = \frac{1+\gamma}{e} - \gamma$ as in previous sections.

\begin{lemma}\label{lem:intersection-after-sstar}
    Any intersection of $F_\gamma(\sigma)$ and $1-\sigma$ can only occur when $\sigma\ge s^*$ for $0<\gamma\le\frac{1}{e-1}$.
\end{lemma}
\begin{proof}
    On $[0,s^*]$, the derivative of $F_\gamma(\sigma)$ with respect to $\sigma$ is
    \begin{align*}
        F_\gamma'(\sigma)=\frac{\gamma}{e(\gamma+\sigma)^2}-\frac{\gamma}{\gamma+1},
    \end{align*}
    which decreases with $\sigma$. The first derivative's value at $s^*$ is
    \begin{align*}
        F_\gamma'(s^*) = \frac{\gamma}{e\left(\frac{1+\gamma}{e}\right)^2}-\frac{\gamma}{\gamma+1} = \frac{\gamma}{(\gamma+1)^2}(e-(\gamma+1))>0,
    \end{align*}
    because $\gamma\le\frac{1}{e-1}$ implies $\gamma+1\le \frac{e}{e-1}<e$. Hence $F_\gamma(\sigma)$ is increasing on $[0,s^*]$ while $1-\sigma$ is decreasing. We will show that $1-s^* > F_\gamma(s^*)$ because this will imply that an intersection can only occur after $s^*$. It remains to compare the two functions at $s^*$. Evaluating $F_\gamma(\sigma)$ at $s^*$ gives us
    \begin{align*}
        F_\gamma(s^*)
        &=\frac1e\left(\frac{s^*}{\gamma+s^*}\right)+\frac{\gamma}{\gamma+1}(1-s^*)\\
        &=\frac{(e-1)\gamma^2+1}{e(\gamma+1)}.
    \end{align*}
    Meanwhile, evaluating $1-\sigma$ at $s^*$ gives us
    \begin{align*}
        1-s^*=(1+\gamma)\left(1-\frac1e\right).
    \end{align*}
    Taking the difference gives
    \begin{align*}
        1-s^*-F_\gamma(s^*)
        =\frac{(e-1)(\gamma+1)^2-1-(e-1)\gamma^2}{e(\gamma+1)}
        =\frac{(e-2)+2(e-1)\gamma}{e(\gamma+1)}>0.
    \end{align*}
    Therefore, the intersection cannot occur before $s^*$.
\end{proof}

\begin{figure}[htbp]
    \centering
    \makebox[\textwidth][c]{%
        \includegraphics[width=1\textwidth]{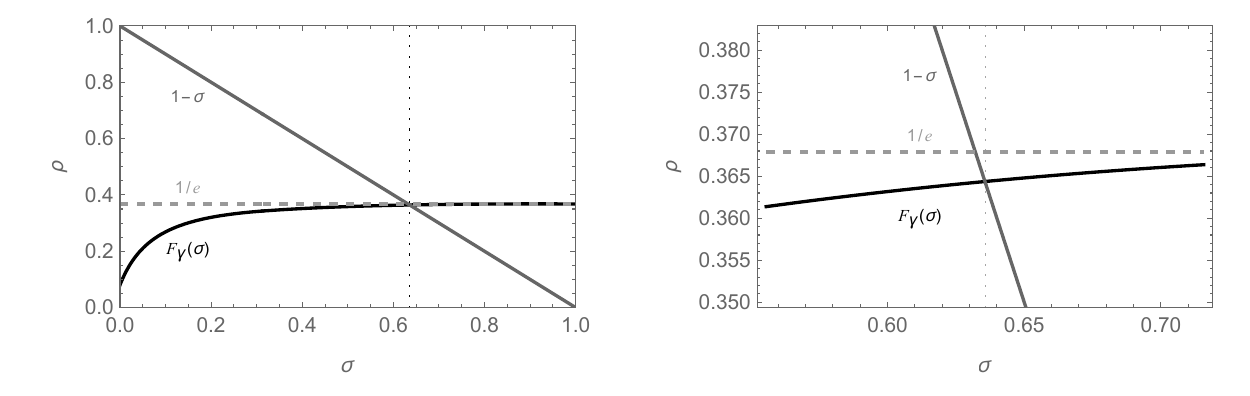}
    }
    \caption{The left plot shows $F_\gamma(\sigma)$ where $\gamma \approx 0.08704$, $1-\sigma$, and $1/e$ over the $[0,1]$ interval, and the right plot zooms in on the intersection point to show it is below $1/e$.}
    \label{fig:intersection-plot}
\end{figure}

\Cref{thm:bound-below-one-over-e} proves our main impossibility result. The idea behind argument can be viewed pictorially in \Cref{fig:intersection-plot}. The proof has two main parts. We show that $F_\gamma(\sigma)$ is increasing up to its intersection with $1-\sigma$, and we also show that this intersection point is below $1/e$.

\begin{theorem}\label{thm:bound-below-one-over-e}
    The optimal competitive ratio for the knapsack secretary problem is at most
    \begin{align*}
        0.36437 < \frac1e-0.0035.
    \end{align*}
\end{theorem}
\begin{proof}
    Consider the point at which $1-\sigma$ and $F_\gamma(\sigma)$ first intersect. It is clear that $1-\sigma$ is decreasing up to the intersection point. If we also show that $F_\gamma(\sigma)$ is increasing up to the intersection point, then $\min\{F_\gamma(\sigma),1-\sigma\}$ is maximized at the first intersection point. We will select $\gamma$ such that $0 < \gamma \leq \frac{1}{e-1}$ and $F_\gamma(\sigma)$ has our desired property. We know from \Cref{lem:intersection-after-sstar} that any intersection point must occur when $\sigma > s^*$ for $0 < \gamma \leq \frac{1}{e-1}$. When $\sigma \leq s^*$, we already know from the proof of \Cref{lem:intersection-after-sstar} that $F'_\gamma(\sigma) > 0$. Thus, we only need to focus on the case when $\sigma > s^*$. We will compute the intersection point by solving for $\sigma$ when $F_\gamma(\sigma) = 1 - \sigma$ and $\sigma > s^*$:
    \begin{align}\label{eq:intersection-equation}
        \frac1e\left(\frac{s^*}{\gamma+s^*}\right)+\frac{\gamma}{\gamma+1}\ln\left(\frac{\gamma+\sigma}{\gamma+s^*}\right)+\frac{\gamma}{\gamma+1}(1-\sigma) & = 1-\sigma.
    \end{align}
    We select $\gamma$ to maximize the intersection point $\sigma^*$, which minimizes the resulting $1-\sigma^*$ value. Solving \eqref{eq:intersection-equation} and optimizing over $\gamma\in(0,\frac{1}{e-1}]$ gives
    \begin{align*}
        \gamma \approx 0.08704,
        \qquad
        \sigma^* \approx 0.63563.
    \end{align*}
    Observe that for this $\gamma$, $s^* \approx 0.31286$, so we do not contradict \Cref{lem:intersection-after-sstar}, and we have that
    \begin{align*}
        F_\gamma'(\sigma)&=\frac{\gamma}{\gamma+1}\left(\frac{1}{\gamma+\sigma}-1\right) > 0
    \end{align*}
    on $(s^*,\sigma^*]$. Recall that one-sided first derivatives also agree at $s^*$. Thus, $F_\gamma(\sigma)$ is increasing up to the intersection point $\sigma^*$ for the given $\gamma$. Combining this with \Cref{cor:converge}, it follows that any algorithm's guarantee can be no more than
    \begin{align*}
        1-\sigma \approx 0.36437.
    \end{align*}
    Since $1/e\approx0.36788$, this is less than $1/e-0.0035$.
\end{proof}

\section{A refined algorithm for the $1$-$B$ knapsack secretary problem}
\label{sec:one-B}
In addition to our impossibility analysis, we also provide an improved approximation for $1$-$B$ knapsack instances. Throughout this section, we write $k:=B$. We use the standard random-order model. Equivalently, each item receives an independent arrival time uniformly distributed in $[0,1]$. All values are assumed to be nonnegative. For notational simplicity, we assume that all item values are distinct; ties may be broken by an arbitrary fixed perturbation.

For an input instance $\mathcal{I}$, let
\[
    \OPT_{\ell}
    :=
    \max\{v_i: i \text{ is large}\}
\]
be the value of the best large item, where the maximum is $0$ if no large item exists, and let
\[
    \OPT_{s}
    :=
    \sum_{i=1}^{k} v_i^{s}
\]
be the value of the optimal packing of the small items, where $v_1^{s} \ge v_2^{s} \ge \cdots $ are the small-item values, padded with zero-valued dummy items if fewer than $k$ small items exist. These items do not change $\OPT_{s}$ and do not affect the distribution of the arrival times of the real items. Since the algorithm only accepts items strictly above the threshold, zero-valued dummy items do not decrease the value obtained from real positive-valued items. Then,
\[
    \OPT = \max\{\OPT_{\ell},\OPT_{s}\}.
\]

\paragraph{The algorithmic template.}

We specialize the sequential large/small framework of Albers, Khan, and Ladewig~\cite{albers2021improved} to the 1-B structure. In this special case, the large-item subproblem collapses to a one-choice secretary problem, and the small-item subproblem is exactly a k-secretary problem with $k=B$. This lets us replace the general-purpose subroutines by sharper secretary subroutines and obtain a better guarantee. For the case $B=2$, we use the boosting algorithm of Abels, Ladewig, Schewior, and Stinzendörfer~\cite{abels2022knapsack}. Their algorithm multiplies the values of the small items by a factor $\alpha>1$ and then applies the size-oblivious secretary algorithm of Albers, Khan, and Ladewig~\cite{albers2021improved}. They show that, for $1.400382 \lesssim \alpha \le e/(e-1)$, this algorithm is $(1/e-o(1))$-competitive for the $1$-$2$ knapsack secretary problem.

For the remaining cases $B \ge 3$, we use the following two-phase algorithm.

\begin{algorithm}[H]
\caption{Three-Phase Algorithm for $1$-$B$ Knapsack, $B\ge 3$}
\label{alg:two-phase-refined}
\KwParams{$0<c_1<c_2<1$}

    Observe and reject all items arriving before time $c_1$; let $r_{\mathrm{sample}}$ be the largest value observed before time $c_1$ with the convention $r_{\mathrm{sample}}=-\infty$ if no item is observed before time $c_1$\;

    \ForEach{item $i$ arriving during $[c_1,c_2)$, in arrival order}{
        \If{$i$ is large and $v_i>r_{\mathrm{sample}}$}{
            Select $i$ and terminate the algorithm\;
        }
    }
    
    Let $\mu$ be the $k$-th largest small-item value observed before time $c_2$, where missing small items are padded by value $0$\;
    
    \ForEach{item $i$ arriving during $[c_2,1]$, in arrival order}{
        \If{$i$ is small, $v_i>\mu$, and fewer than $k$ small items have been selected}{
            Select $i$\;
        }
    }
\end{algorithm}

The zero padding in the definition of $\mu$ is only a notational convention. It means that if fewer than $k$ small items have appeared by time $c_2$, then the threshold is $0$. Since all values are nonnegative, this convention can only help the algorithm.

\subsection{Analysis of the two-phase algorithm for $B\ge 3$}

We use the following special case of the $\singleref$ algorithm of Albers and Ladewig~\cite{albers2021new}. The $\singleref$ algorithm with parameters $(c,r)$ samples the first $c$ fraction of the input, uses the $r$-th best sampled item as a threshold, and then accepts the first $k$ later items beating that threshold. We only use the case $r=k$, which is sufficient for our uniform bound, even though this is not optimal for every $k$ for $\singleref$.

\begin{lemma}
\label{lem:single-ref-r-k}
For the $k$-secretary problem, $\singleref$ with sampling fraction $c$ and reference rank $r=k$ is $(\rho_k(c)-o(1))\text{-competitive}$, where
\[
    \rho_k(c)
    :=
    \int_c^1
    \Pr\!\left[
        \operatorname{Bin}\!\left(2k-1,\frac{c}{t}\right)\ge k
    \right]dt.
\]
\end{lemma}

\begin{proof}
In the notation of Albers and Ladewig~\cite{albers2021new}, let $p_1^{(j)}$ be the probability that the best item is accepted as the $j$-th selected item. Theorem~1 of~\cite{albers2021new} gives the competitive ratio of $\singleref$ as
\[
    \frac{1}{k}
        \sum_{j=1}^{k} \gamma_j \cdot p_1^{(j)}
    -o(1), \qquad
    \text{where}\ \gamma_j = \begin{cases}
        r + 2(j-1),& \text{if}\ j \le k-r+1\\
        k,& \text{else}.
    \end{cases}
\]
For $r=k$, $\gamma_j = k$ for all $j$, and the expression reduces to
\[
    \sum_{j=1}^{k}p_1^{(j)}-o(1),
\]
which is exactly the probability that the best item is accepted, up to the asymptotic error term.

It remains to express this probability. Condition on the best item arriving at time $t>c$. Up to the asymptotically negligible boundary effects absorbed in the $o(1)$ term, with reference rank $r=k$, the best item is accepted exactly when the threshold at time $c$ is below the best item and fewer than $k$ items better than the threshold have already been accepted before it arrives. Equivalently, among the $2k-1$ highest-value items that arrive before time $t$, excluding the best item itself, at least $k$ lie in the sample interval $[0,c]$. Conditional on arriving before time $t$, each of these items lies in the sample interval $[0,c]$ with probability $c/t$, independently. Therefore the acceptance probability of the best item is
\[
    \int_c^1
    \Pr\!\left[
        \operatorname{Bin}\!\left(2k-1,\frac{c}{t}\right)\ge k
    \right]dt.
\]
\end{proof}

\begin{lemma}
\label{lem:two-phase-refined}
For every $k=B\ge 3$ and every $0<c_1<c_2<1$, Algorithm~\ref{alg:two-phase-refined} has competitive ratio at least
\[
    \min\left\{
        c_1\ln\frac{c_2}{c_1},
        \;
        \frac{c_1}{c_2}\rho_k(c_2)
    \right\}
    -o(1).
\]
\end{lemma}

\begin{proof}
We distinguish the two possible forms of the offline optimum.

\paragraph{Case 1: $\OPT=\OPT_{\ell}$.}

Let $g^\star$ be the most valuable large item. Since the optimal solution is a large item, no small item has value larger than $v_{g^\star}$; otherwise that single small item would already give a better feasible solution. Thus $g^\star$ is also a globally maximum-value item.

We lower-bound the probability that Algorithm~\ref{alg:two-phase-refined} selects $g^\star$.  Condition on $g^\star$ arriving at time $t\in[c_1,c_2)$. If no large item arrives before time $t$, then no large item can be selected before $g^\star$, and the algorithm selects $g^\star$ at time $t$. Otherwise, let $h_t$ be the highest-value large item arriving before time $t$. If $h_t$ arrives before time $c_1$, then $r_{\mathrm{sample}}\ge v_{h_t}$. Consequently, no large item arriving in $[c_1,t)$ has value larger than $r_{\mathrm{sample}}$, so no large item is selected before time $t$. Moreover, since $g^\star$ is the globally best item and was not sampled, we have $v_{g^\star}>r_{\mathrm{sample}}$, and hence the algorithm selects $g^\star$ when it arrives.

Conditional on the set of large items arriving before time $t$ being nonempty, the arrival time of its highest-value member is uniformly distributed in
$[0,t]$. Therefore,
\[
    \Pr[g^\star\text{ is selected}\mid T_{g^\star}=t]
    \ge
    \frac{c_1}{t}.
\]
Integrating over $t\in[c_1,c_2)$ gives
\[
    \Pr[g^\star\text{ is selected}]
    \ge
    \int_{c_1}^{c_2}\frac{c_1}{t}\,dt
    =
    c_1\ln\frac{c_2}{c_1}.
\]
Thus
\[
    \mathbb E[\ALG]
    \ge
    c_1\ln\frac{c_2}{c_1}\cdot \OPT.
\]

\paragraph{Case 2: $\OPT=\OPT_{s}$.}

Let $t_i$ represent the arrival time of item $i$. Let $L_{<c_2} := \{i : i \text{ is large and } t_i<c_2\}$ be the set of large items arriving before time $c_2$. We define an event $H$ that guarantees that the algorithm reaches the small-item phase. If $L_{<c_2}=\emptyset$, then let $H$ occur. Otherwise, let $g_{<c_2} \in \arg\max\{v_i : i \in L_{<c_2}\}$
be the highest-value large item arriving before time $c_2$, and define $H:=\{t_{g_{<c_2}}<c_1\}$.

We first lower-bound the probability of $H$. Conditional on $L_{<c_2}\neq\emptyset$ and on the identity of the set $L_{<c_2}$, the arrival times of the items in $L_{<c_2}$ are independent and uniform in $[0,c_2)$. Hence the arrival time of the highest-value item in this set is uniform in $[0,c_2)$. Since this probability is $c_1/c_2$ for every fixed nonempty identity of $L_{<c_2}$, averaging over the possible identities gives
\[
    \Pr[H\mid L_{<c_2}\neq\emptyset]
    =
    \frac{c_1}{c_2}.
\]
If $L_{<c_2}=\emptyset$, then $H$ occurs by definition. Thus $\Pr[H]\ge \frac{c_1}{c_2}$.

On the event $H$, the algorithm cannot select a large item before time $c_2$. Indeed, if $L_{<c_2}=\emptyset$, then no large item arrives before $c_2$.  Otherwise, $g_{<c_2}$ is contained in the first sample, and hence for every $i \in L_{<c_2}$:
\[
    r_{\mathrm{sample}}
    \ge
    v_{g_{<c_2}}
    \ge
    v_i.
\]
Therefore, no large item arriving in $[c_1,c_2)$ has value larger than $r_{\mathrm{sample}}$, and the algorithm reaches the small-item phase.

The event $H$ depends only on the arrival times of large items, and is independent of the arrival times and relative order of the small items.  Thus, conditional on $H$, the distribution of the small-item arrival sequence is unchanged. Moreover, conditional on $H$, the small-item phase of Algorithm~\ref{alg:two-phase-refined} is exactly the $\singleref$ $k$-secretary algorithm with sampling fraction $c_2$ and reference rank $r=k$. By Lemma~\ref{lem:single-ref-r-k},
\[
    \mathbb E[\ALG\mid H]
    \ge
    \bigr(\rho_k(c_2)-o(1)\bigl) \OPT_{s}.
\]
Since the algorithm obtains a nonnegative value on the complement of $H$, we get
\[
\begin{aligned}
    \mathbb E[\ALG]
    &\ge
    \Pr[H]\cdot \mathbb E[\ALG\mid H]  \\
    &\ge
    \frac{c_1}{c_2}
    \bigl(\rho_k(c_2)-o(1)\bigr)\OPT_{s}  \\
    &=
    \left(\frac{c_1}{c_2}\rho_k(c_2)-o(1)\right)
    \OPT_{s}.
\end{aligned}
\]
Since $\OPT=\OPT_{s}$ in the present case, this proves
\[
    \mathbb E[\ALG]
    \ge
    \left(\frac{c_1}{c_2}\rho_k(c_2)-o(1)\right)\OPT.
\]

Combining the two cases proves the lemma.
\end{proof}

We now choose explicit parameters. The choice is optimized for the worst case $k=3$; larger values of $k$ only improve the small-item term.

\begin{lemma}
\label{lem:rho-monotone}
For every $c\ge 1/2$, the sequence $\rho_k(c)$ is non-decreasing in $k$. Consequently, for every $k \ge 3$ and every $c \ge 1/2$,
\[
    \rho_k(c)\ge \rho_3(c).
\]
\end{lemma}

\begin{proof}
For $q\in[0,1]$, define $M_k(q) := \Pr[\operatorname{Bin}(2k-1,q)\ge k]$. We first show that $M_k(q)$ is non-decreasing in $k$ whenever $q\ge 1/2$. Let $X \sim \operatorname{Bin}(2k-1,q)$, and let $Y_1,Y_2$ be independent Bernoulli random variables with success probability $q$, independent of $X$. Then
\[
    M_{k+1}(q)
    =
    \Pr[X+Y_1+Y_2\ge k+1].
\]
Thus
\[
\begin{aligned}
    M_{k+1}(q)-M_k(q)
    &=
    q^2\Pr[X=k-1]-(1-q)^2\Pr[X=k].
\end{aligned}
\]
Since $\Pr[X=k] = \frac{q}{1-q}\Pr[X=k-1]$, we obtain
\[
\begin{aligned}
    M_{k+1}(q)-M_k(q)
    &=
    \left(q^2-q(1-q)\right)\Pr[X=k-1]\\
    &=
    q(2q-1)\Pr[X=k-1].
\end{aligned}
\]
Therefore $M_k(q)$ is non-decreasing in $k$ for $q\ge 1/2$.

Now suppose $c\ge 1/2$. For every $t\in[c,1]$, $\frac{c}{t} \ge c \ge \frac{1}{2}$. Hence, for all $k \ge 3$:
\[
    M_k\!\left(\frac{c}{t}\right)
    \ge
    M_3\!\left(\frac{c}{t}\right).
\]
Integrating over $t\in[c,1]$ gives
\[
    \rho_k(c)
    =
    \int_c^1 M_k\!\left(\frac{c}{t}\right)dt
    \ge
    \int_c^1 M_3\!\left(\frac{c}{t}\right)dt
    =
    \rho_3(c).
\]
\end{proof}

For $k=3$, we can compute $\rho_3(c)$ explicitly:
\[
\begin{aligned}
    \rho_3(c)
    &=
    \int_c^1
    \Pr\!\left[
        \operatorname{Bin}\!\left(5,\frac{c}{t}\right)\ge 3
    \right]dt  \\
    &=
    \int_c^1
    \left(
        10\left(\frac{c}{t}\right)^3 \left(1-\frac{c}{t}\right)^2
        +5\left(\frac{c}{t}\right)^4 \left(1-\frac{c}{t}\right)
        +\left(\frac{c}{t}\right)^5
    \right)dt\\
    &=
    \int_c^1
    \left(
        10\left(\frac{c}{t}\right)^3
        -15\left(\frac{c}{t}\right)^4
        +6\left(\frac{c}{t}\right)^5
    \right)dt\\
    &=
    \frac32 c-5c^3+5c^4-\frac32 c^5.
\end{aligned}
\]

\begin{theorem}
\label{thm:two-phase-B-ge-3}
For every fixed $B=k\ge 3$, Algorithm~\ref{alg:two-phase-refined} with $ c_1=0.3657, c_2=0.6258$ is at least $(1/5.10-o(1))$-competitive.
\end{theorem}

\begin{proof}
By Lemma~\ref{lem:two-phase-refined}, the competitive ratio is at least
\[
    \min\left\{
        c_1\ln\frac{c_2}{c_1},
        \;
        \frac{c_1}{c_2}\rho_k(c_2)
    \right\}
    -o(1).
\]

First,
\[
    c_1\ln\frac{c_2}{c_1}
    =
    0.3657\ln\frac{0.6258}{0.3657}
    >
    0.19646
    >
    \frac1{5.10}.
\]

Second, since $c_2=0.6258>1/2$, Lemma~\ref{lem:rho-monotone} gives
$\rho_k(c_2)\ge \rho_3(c_2)$ for every $k\ge3$.  Using the closed form for
$\rho_3$,
\[
    \rho_3(0.6258)
    =
    \frac{3}{2}(0.6258)
    -5(0.6258)^3
    +5(0.6258)^4
    -\frac{3}{2}(0.6258)^5
    >
    0.33618.
\]
Therefore
\[
    \frac{c_1}{c_2}\rho_k(c_2)
    \ge
    \frac{0.3657}{0.6258}\cdot 0.33618
    >
    0.19645
    >
    \frac1{5.10}.
\]
Thus both terms in the minimum are larger than $1/5.10$, and the theorem
follows.
\end{proof}

\subsection{Combining with the $1$-$2$ algorithm}

We now combine the two-phase algorithm above with the specialized algorithm of Abels, Ladewig, Schewior, and Stinzendörfer~\cite{abels2022knapsack} for $B=2$.

\begin{algorithm}[H]
\caption{Combined Algorithm for $1$-$B$ Knapsack}
\label{alg:combined-one-B}
\setcounter{AlgoLine}{0}
    \KwParams{capacity parameter $B\ge 2$}
    
    \eIf{$B=2$}{
        Run the boosted $1$-$2$ knapsack secretary algorithm of
        Abels et al.~\cite{abels2022knapsack} with any
        $\alpha\in[1.400382,e/(e-1)]$\;
    }{
        Run Algorithm~\ref{alg:two-phase-refined} with
        $c_1=0.3657$ and $c_2=0.6258$\;
    }
\end{algorithm}

\begin{theorem}
\label{thm:combined-one-B}
For every fixed $B\ge 2$, Algorithm~\ref{alg:combined-one-B} is $(1/5.10-o(1))$-competitive for the $1$-$B$ knapsack secretary problem.
\end{theorem}

\begin{proof}
If $B=2$, the guarantee follows from Abels et al.~\cite{abels2022knapsack}, whose algorithm is $(1/e-o(1))$-competitive. Since $1/e>1/5.10$, this implies the claimed bound. If $B \ge 3$, the guarantee follows from Theorem~\ref{thm:two-phase-B-ge-3}.
\end{proof}

\paragraph{Comparison with the general knapsack bound.}

The general random-order online knapsack algorithm of Albers, Khan, and Ladewig~\cite{albers2021improved} is $(1/6.65-o(1))$-competitive but works for more general instances of knapsack secretary. In comparison, Theorem~\ref{thm:combined-one-B} gives a better guarantee specifically for the $1$-$B$ knapsack secretary problem when $B \ge 2$.

\paragraph{Remark.} For large $B$, Abels et al.~\cite{abels2022knapsack} give a stronger ordinal guarantee of $(1/(e+1)-o(1))$. Our result is complementary: it gives a simple uniform guarantee for every fixed $B \ge 2$, with the case $B=2$ handled by their boosted algorithm and all cases $B \ge 3$ handled by the two-phase algorithm.
\section{Conclusion and Future Directions}

Our hard instance construction shows that the classical \(1/e\) benchmark is not the right barrier for knapsack secretary, even for cardinal algorithms. The construction is deliberately simple: it already suffices to consider \(1\)-\(B\) instances in which each item either consumes one unit of capacity or the entire knapsack. Within this restricted family, the hard distribution is supported only on the number of large items, and it exposes the central tension of the problem: the algorithm must wait long enough to compete for the best large item, while not waiting so long that it loses too much value in the case where the optimum consists of small items.

A first direction is to strengthen the quantitative impossibility result obtained from this framework. In our proof, the beta-binomial prior is chosen for analytic tractability rather than because it is extremal. Its conjugacy leads to a low-dimensional dynamic program and, in the scaling limit, to an explicitly solvable ODE, which is enough to certify a competitive ratio strictly below \(1/e\). Numerical evaluation of the finite dynamic program suggests that this analytic certificate is not tight. Already for moderate finite parameter values, empirically optimizing the distribution over the number of large items appears to give a slightly stronger impossibility result. Turning this numerical evidence into a clean analytic impossibility result remains an interesting open problem. Such a proof would likely require either a more precise characterization of the extremal prior or a different relaxation that retains enough structure to be analyzed in closed form.

A second direction is to enrich the small-item side of the hard instances. In the present construction, all small items have the same value, and the knapsack can accommodate all of them whenever the no-large-item instance occurs. This makes the small-item branch behave essentially as a deterministic fallback: before the first large item appears, the algorithm is balancing secretary-style exploration against a single deadline, and the optimal policy does not commit to small items before the \(1-1/e\) time scale. More refined hard instances could assign heterogeneous values to the small items and allow the knapsack to contain only a limited number of them. This would introduce a genuine online selection problem on the small-item branch, rather than a uniform residual-value tradeoff, and may further decrease the achievable competitive ratio.

Finally, it would be interesting to extend the technique beyond knapsack secretary. The proof combines three ingredients: a structured hard family, a distributional relaxation that strengthens the algorithm, and an ordinalization step showing that cardinal information is asymptotically irrelevant on the chosen instances. Similar conflicts between committing to one type of feasible solution and preserving the option of another arise in other combinatorial secretary problems, including edge-arrival matching and matroid secretary. Understanding whether analogous constructions can break the inherited \(1/e\) impossibility barrier in these settings is a natural next step.

\printbibliography
% \appendix

\end{document}